\documentclass[a4paper,USenglish,numberwithinsect]{lipics-v2018}
\usepackage{tikz}
\usepackage{cleveref}

\usetikzlibrary{backgrounds,calc,decorations.pathreplacing, hobby}

\def\Xdef#1{\expandafter\xdef\csname#1\endcsname}
\def\TakeInt#1#2{\GetTake#2\Fin\expandafter\PutTakeInt\pgfmathresult.0:#1\Fin}
\def\GetTake#1\Fin{\pgfmathparse{#1}\pgfmathparse{int(\pgfmathresult)}}
\def\PutTakeInt#1.#2:#3\Fin{\Xdef{Result#3}{#1}}

\newcommand{\mN}{\mathbb{N}}
\def\cost{\mbox cost}
\def\En{E_0}
\def\Ee{E_1}
\def\Ex{E_{Multi}}
\def\Eu{E_{used}}

\def\KnotenDiam{0.9mm}

\title{The Graph Exploration Problem with Advice}

\author{Hans-Joachim B\"ockenhauer}{ETH Z\"urich\\{Department of Computer Science}}{hjb@inf.ethz.ch}{}{}
\author{Janosch Fuchs}{RWTH Aachen University\\{Computer Science 1}}{fuchs@algo.rwth-aachen.de}{}{[funding:UnRaVeL?]}
\author{Walter Unger}{RWTH Aachen University\\{Computer Science 1}}{quax@algo.rwth-aachen.de}{}{}
\authorrunning{H.-J. B\"ockenhauer and J. Fuchs and W. Unger}
\Copyright{Hans-Joachim B\"ockenhauer and Janosch Fuchs and Walter Unger}
\subjclass{Online algorithms, Graph algorithms analysis}

\nolinenumbers
\begin{document}
\maketitle
\begin{abstract}
Moving an autonomous agent through an unknown environment is one of the crucial problems for robotics and network analysis. 
Therefore, it received a lot of attention in the last decades and was analyzed in many different settings. 
The \emph{graph exploration problem} is a theoretical and abstract model, where an algorithm has to decide how the agent, also called \emph{explorer}, moves through a network such that every point of interest is visited at least once. 
For its decisions, the knowledge of the algorithm is limited by the perception of the explorer.

There are different models regarding the perception of the explorer. 
We look at the \emph{fixed graph scenario} proposed by Kalyanasundaram and Pruhs (Proc. of ICALP, 1993), where the explorer starts at a vertex of the network and sees all reachable vertices, their unique names and their distance from the current position. 
Therefore, the algorithm recognizes already seen vertices and can adapt its strategy during exploring, because it does not forget anything.

Because the algorithm only learns the structure of the graph during computation, it cannot deterministically compute an optimal tour that visits every vertex at least once without prior knowledge. 
Therefore, we are interested in the amount of crucial a-priori information needed to solve the problem optimally, which we measure in terms of the well-studied model of \emph{advice complexity}. 
Here, the algorithm can at any time access a binary string of information written beforehand by an oracle that knows the optimal solution, the network and the behavior of the algorithm. 
The number of bits read by the algorithm until the end of computation, is called the \emph{advice complexity}. 

We look at different variations of the graph exploration problem and distinguish between directed or undirected edges, cyclic or non-cyclic solutions, unit costs or individual costs for the edges and different amounts of a-priori structural knowledge of the explorer. 
For general graphs, it is known that $\mathcal{O}(n\log(n))$ bits of advice suffice to compute an optimal solution. 
In this work, we present algorithms with an advice complexity of $\mathcal{O}(m+n)$, thus improving the classical bound for sparse graphs. 
\end{abstract}



\section{Introduction}
Orientation and navigation in an unknown environment is one of the basic tasks for autonomous agents. 
The environment can be physical or virtual like in a network of computers, where the connections between the computers are unknown. 
To send messages as fast as possible, it is helpful to know the structure of the network. 
Thus an explorer can be used to visit and test the connections for every computer in the network. 

A more physical example is the field of robotics. 
There are many applications for robots that explore unknown environments on their own \cite{thrun2002robotic,thrun2004autonomous,kortenkamp1998artificial}. 
For example, exploring caves or abandoning mines is often a dangerous task and autonomous robots can be used to create maps of such environments. 

Because there are many different applications for exploring, there are also many different models for the environment and for the perception of the explorer. 
The survey of Berman \cite{berman1998line} gives an overview of navigation problems and distinguishes the following main properties: 
The representation of the environment, the task that should be solved, and the senses of the explorer. 
The environment can be a geometric space with obstacles \cite{blum1997navigating} or, this is the case that we analyze in this paper, an abstract and discrete space, where the explorer can move from one point to a neighboring one, i.e., a graph consisting of vertices and edges. 
Starting from a vertex in the given graph, the task can be to find a shortest path to a target vertex or to compute a shortest tour that visits every vertex at least once. 

The task to compute a shortest tour that visits every vertex at least once is related to the well-known \emph{Traveling Salesman Problem}. 
Kalyanasundaram et al.~introduced the graph exploration problem as an online version of the TSP \cite{kalyanasundaram1993constructing} with the \emph{fixed graph scenario} which defines the senses of the explorer. 
In this model, the vertices have unique labels, and the explorer sees all reachable vertices, their labels and their distances to the current position. 
Moving onto a new vertex reveals the adjacent vertices and the explorer recognizes if a vertex was already reachable from a previous step. 

Obviously, not having complete knowledge of the graph beforehand makes it impossible in general for the explorer to find a tour of optimal length. 
Thus, algorithms achieving some provable approximation guarantees have been investigated. 
The best known lower bound on the approximation ratio for exploring general and undirected graphs in the fixed graph scenario is $\frac{5}{2} - \epsilon$ \cite{Dobrev2012}. 
For the special case of undirected graphs with bounded genus $g$, an upper bound of $16(1+2g)$ is known \cite{megow2012online}. 
The case of directed graphs in fixed graph scenario is also well studied \cite{albers2000exploring,fleischer2005exploring,foerster2016lower}. 
In \cite{foerster2016lower}, the authors give tight bounds for deterministic and randomized graph exploration in directed graphs with wighted or unweighted edges. 
Moreover, they look at a variation of the problem where the explorer has to search for a specific vertex in the graph. 
There are many slight variations of the graph exploration problem. 
The memory of the algorithm \cite{diks2004tree,fraigniaud2004digraphs}, the number of explorers in the graph \cite{brass2014improved,chalopin2010network,das2007map} or the abilities to set pebbles \cite{bender2002power} are well studied variations. 
 
For a more fine-grained analysis of how much information about the unknown graph is really needed by the explorer, we look at a variation of the graph exploration problem where the algorithm has access to some information in the form of a bit string, provided by a helpful oracle that knows the network. 
The number of bits that the algorithm reads until it finishes its computation is then called its \emph{advice complexity}. 
We are interested in how many bits are needed to compute an optimal exploration sequence. 
Dobrev et al.~\cite{Dobrev2008} introduced this model in the setting of
online algorithms, which was later improved by Hromkovi\v{c} et al.~\cite{HKK10} and B\"ockenhauer et al.~\cite{Bockenhauer2009} as well as by Emek et al.~\cite{Emek2011online}. 
The setting and many results are explained in detail in \cite{Komm16}. 
We use the model of B\"ockenhauer et al.~\cite{Bockenhauer2009} in this paper. 

The first time that the graph exploration problem was analyzed using advice complexity model was in \cite{fraigniaud2008tree}, where Fraigniaud et al.~were able to improve the classical upper bound of $2$ on the competitive ratio for tree exploration by adding advice. 
They proved that $\log\log(D)-c$ bits of advice suffice for $c$-competitiveness, where $D$ is the diameter of the input tree and $c<2$. 
Moreover, they showed that every algorithm that uses less advice bits has to be at least $2$-competitive. 

Since then, there have been many results regarding graph searching problems with advice \cite{Dobrev2012,komm2015treasure,gorain2016deterministic}. 
The search for a specific vertex in the graph stands in focus of research in \cite{komm2015treasure}. 
The authors present an algorithm that uses $\Theta (n/r)$ bits of advice for a competitive ratio $r$. 
In \cite{Dobrev2012}, Dobrev et al.~looked at the trade-off between advice and competitiveness for the cyclic graph exploration problem. 
More over they show that $\Omega (n\log n)$ bits of advice are necessary for optimality. 
Gorain et al.~\cite{gorain2016deterministic} show bounds for a weaker oracle model, where the oracle does not know the starting position of the algorithm. 
Moreover, they show that a number of advice bits linear in the number of vertices results in a solution of quadratic size. 

In the following, we will look at the graph exploration problem for different settings and prove that $\mathcal{O} (m+n)$ bits of advice suffice to compute an optimal solution. 
Note that an upper bound of $\mathcal{O}(n\log n)$ advice bits can be easily achieved by sorting the vertices by their first visits and encoding this order. 
Our result improves over this strategy for sparse graphs. 
We first concentrate on the case that the graph is directed edge-weighted and unknown to the algorithm, and the goal is to compute a cyclic tour visiting all vertices. 
After that, we show how all other variants can be solved by some modification of the algorithm. 
Among the different settings, the a-priori knowledge of the input graph has the largest impact on the advice complexity. 
Since our algorithm relies on the encoding of an optimal solution within the given advice and does not take the edge costs into account, we formulate our results for the more general case of arbitrary edge weights only. 

The paper is organized as follows. 
In Section \ref{sec:definitions}, we give the basic definitions for dealing with the graph exploration problem. 
Section \ref{sec:basic} gives some basic observations and, in Section \ref{sec:known} we develop some techniques to compute an optimal solution for the case where the graph is known to the algorithm. 
These techniques are used in the more complicated case where the input graph is unknown. 
In Section \ref{sec:bounded}, we prove our main result for directed graphs with in-degree and out-degree bounded by $2$. 
This result is extended to directed graphs of arbitrary degree in Section \ref{sec:general}, and, in Section \ref{sec:extensions}, we consider extensions of our algorithm to the cases of undirected graphs and path solutions instead of cyclic solutions. 

The results for all considered variants are summarized in \Cref{table:results}. 

\begin{table}[t]
  \caption{Upper bounds on the advice complexity of graph
    exploration}\label{table:results} 
  \begin{center}
    \begin{tabular}{lll}
      \textbf{graph model} & \textbf{exploration} & \textbf{upper bound}\\
      unknown directed & cyclic & $2n+23m$\\
      unknown directed & path & $2n+23m + \lceil\log(n)\rceil$\\
      unknown undirected & cyclic & $\log(6)(n+m) + 42m$\\
      unknown undirected & path & $\log(6)(n+m) + 42m + \lceil\log(n)\rceil$ \\
      known directed & cyclic & $\lceil\log(3) m\rceil$\\
      known directed & path & $\lceil\log(3) m\rceil+\lceil\log(n)\rceil$\\
      known undirected & cyclic & $\lceil\log(6) m\rceil$\\
      known undirected & path & $\lceil\log(6) m\rceil+\lceil\log(n)\rceil$
    \end{tabular}
  \end{center}
\end{table}


\section{Basic Definitions}
\label{sec:definitions}
We start with a definition of the basic variant of the graph exploration
problem that we consider in this paper. 

\begin{definition}\label{def-variations}
  Let $G=(V,E)$ be a directed graph. Every vertex $v\in V$ has a fixed unique identifier. 
  There is an agent, called \emph{explorer}, initially positioned on some start vertex $v_0 \in V$. 
  The algorithm has to move this explorer along the directed edges of $G$ to visit all vertices and return to $v_0$. 
  For the non-cyclic graph exploration problem, the explorer does not need to return to $v_0$. 
  The edges are weighted by a cost function $\cost \colon E \to \mN$, and the goal is to minimize the total cost along the cyclic tour traveled by the explorer. 
  In every vertex the explorer is located, it sees the outgoing edges, their costs, and the vertex identifiers at the endpoints of these edges, but not the incoming edges.
\end{definition}

Throughout the paper, we denote, for any graph $G=(V,E)$ with vertex set $V$ and edge set $E$, the number of vertices by $n=|V|$ and the number of edges by $m=|E|$. 
For the graph exploration problem, the algorithm decides in each round along which edge the explorer is moved. 
So, starting at the vertex $v_0$, the explorer sees the reachable neighborhood of the starting vertex and their unique identifiers. 
In the following, we will call the reachable neighborhood the \emph{out-neighborhood}, which is defined as follows. 

\begin{definition}
Let $G=(V,E)$ be a directed graph. 
The \emph{out-neighborhood} of a vertex $v \in V$ is defined as $N_{out} (v) = \{ w \mid (v,w) \in E \}$. 
Analogously, we define the \emph{in-neighborhood} of a vertices $v$ as $N_{in} (v) = \{ w \mid (w,v) \in E \}$. 
\end{definition}

We describe the tour followed by the explorer in terms of a \emph{search sequence}. 

\begin{definition}
Let $G=(V,E)$ be a graph, a sequence $S=(v_0,v_1, \ldots , v_{s})$ is called a \emph{search sequence} if $(v_{i-1},v_i)\in E$ for all $1\leq i\leq s$. 
If $v_0 = v_{s}$, we call $S$ a \emph{cyclic search sequence}. 

For a search sequence $S=(v_0, \ldots, v_k)$ with $e_i = (v_{i-1}, v_i)$, for $1 \leq i \leq k$, we denote by $E(S)=(e_1,e_2,\ldots,e_k)$ the sequence of edges in $S$ and by $V(S)=\{v_0,\ldots,v_k\}$ the set of vertices in $S$. 
We also interpret $E(S)$ as a multiset of edges and thus write $e\in E(S)$ if there exists some $i$ with $e=e_i$. 
The cost of a search sequence $S=(v_0, \ldots,v_k)$ is defined by $\cost(S)=\sum_{e\in E(S)}\cost(e)\cdot \#_{S}(e)$, where $\#_{S}(e)$ is the number of traversals through an edge $e\in E$ in the search sequence $S$. 
\end{definition}

The search sequence is determined by the algorithm as follows. 
In each step, the explorer is located at some vertex $v$ and the algorithm chooses one of the vertices from $N_{out}(v)$ as target and moves the explorer towards it. 
As soon as the explorer arrives at a vertex, a new round starts and the algorithm receives again the unique identifiers for the reachable neighborhood and has to make an unrecoverable decision, the \emph{response}. 
With each round, the algorithm extends its search sequence by one. 
The cost for a sequence is the number of used edges or, if a cost function is given, the sum of the costs for the used edges. 
The goal is to compute a cyclic search sequence visiting each vertex at least once, we call such a sequence an \emph{exploration sequence}. 

\begin{definition}
A search sequence $S=(v_0,v_1, \ldots ,v_{end})$ for a graph $G=(V,E)$ is called an \emph{exploration sequence} if $V(S)=V$ and $v_0$ is the start vertex. 
If $v_0 = v_{end}$, we call $S$ a \emph{cyclic exploration sequence}. 
\end{definition}

Note that with each decision, the algorithm influences the new input for the next decision. 
Thus, strictly speaking, the graph exploration problem is no classical online problem. 
But the adversary still knows the behavior of the deterministic algorithm and can, with this knowledge, prepare the input graph, the unique identifiers for the vertices, and thus the enumeration of the edges. 
Hence, we can analyze the graph exploration problem using the same methodology as used for online problems. 



Since the algorithm lacks global information about the structure of the graph, there is no deterministic algorithm that finds an optimal exploration sequence for any graph. 
We employ the model of online algorithms with advice as defined in \cite{HKK10,Bockenhauer2009} for measuring the amount of missing information, which we can define in the framework of graph exploration as follows.

An \emph{online algorithm with advice} computes a search sequence $S=(v_0,v_1, \ldots ,v_{end})$, where $v_i$ is computed from $N_{out}(v_0),\ldots, N_{out}(v_{i-1})$ (i.e., from the partial knowledge about the graph gathered in the first $i-1$ rounds) and the content $\phi$ of the advice tape, i.e., an unbounded binary sequence of \emph{advice bits} computed by an \emph{oracle} that sees the complete input graph together with its edge-cost function. 
An online algorithm with advice \emph{solves} the graph exploration problem if, for any input $(G,\cost)$, there exists a computable advice $\phi$ such that $S$ is an optimal exploration sequence. 
The algorithm has sequential access to the bits from the advice tape, and its \emph{advice complexity} is the number of accessed advice bits. 
As usual, we measure the advice complexity with respect to the input size by considering a worst-case input of the respective size. 

We now define how the algorithm makes a decision to extend a search sequence $S$, based on some advice from the oracle. 
The basic idea is that the oracle chooses some fixed optimal exploration sequence $S^*$ and communicates a sufficient amount of information about it such that the algorithm can compute this sequence without taking the costs of the edges into consideration. 
We observe that it is sufficient for the algorithm to know the exact number of traversals for each edge in $E(S^*)$. 
Explicitly communicating these traversal numbers could be done with $\mathcal{O}(m\log n)$ advice bits in a straightforward way since no edge can be traversed more than $n$ times in an optimal exploration sequence. 
But, for sparse graphs, this would be too expensive, so the oracle only communicates some partial information from which the algorithm can compute the traversal numbers. 
As a first step, we partition the edges into three sets according to their number of traversals (none, one, or multiple times) in an optimal exploration sequence.

\begin{definition}
  Let $G=(V,E)$ be a graph, let $S^*=(v_0, \ldots ,v_{end})$ be an exploration sequence of minimum cost and let $S=(v_0, \ldots ,v_k)$ be an arbitrary search sequence. 
  We define:
  \begin{itemize}
  \item $\En = \{ e \in E \mid \#_{S^*}(e)=0\}$ is the set of edges in $E$ which are not visited by $S^*$. 
  \item $\Ee = \{ e \in E \mid \#_{S^*}(e)=1\}$ is the set of edges in $E$ which are visited exactly once by $S^*$. 
  \item $\Ex = \{ e \in E \mid \#_{S^*}(e)>1\}$ is the set of edges in $E$ which are visited more than once by $S^*$. 
  \item $\Eu = \Ex \cup \Ee$ is the set of edges in $E$ which are visited at least once by $S^*$. 
  \end{itemize}
  If $\#_{S}(e) = \#_{S^*}(e)$ holds, e.g., an edge $e$ is as often used in $S$ as in the fixed optimal solution $S^*$, we say $e$ is \emph{exhausted} in $S$. 
\end{definition}



The number of traversals for the edges could differ for different optimal solutions, but the oracle fixes one of the optimal solutions such that the given advice is consistent during the exploration. 
\Cref{unique-solution} shows a sample graph where the number of traversals for the edges in an optimal solution is non-unique. 
The five traversals over the vertex $x$ needs to be split up between the two paths $(y,v_1,x)$ and $(y,v_2,x)$. 

\begin{figure}[htb]
    \begin{center}
    \begin{tikzpicture}[>=latex]
    \tikzset{line/.style={-latex'}}
    \draw (0, 0) node[label={[label distance=-0.03cm]180:$y$}] (y) {}; \filldraw (y) circle (\KnotenDiam);
      \draw (1, -0.8) node[label={[label distance=-0.03cm]270:$v_1$}] (r) {}; \filldraw (r) circle (\KnotenDiam);
      \draw (1, 0.8) node[label={[label distance=-0.03cm]270:$v_2$}] (l) {}; \filldraw (l) circle (\KnotenDiam);
      \draw (2,0) node[label={[label distance=-0.05cm]90:$x$}] (x) {}; \filldraw (x) circle (\KnotenDiam);
      \draw (5.5,0) node (c) {}; \filldraw (c) circle (\KnotenDiam);
      \foreach \x in {0,...,4} {\draw (4.5,0.85*\x-1.5) node (\x) {}; \filldraw (\x) circle (\KnotenDiam);
        \draw[->=latex,line width=1] (x) -- (\x) {}; \draw[->=latex,line
        width=1] (\x) -- (c) {}; }
      \draw[->,line width=1] (y) -- (r) {}; \draw[->,line width=1] (y) -- (l) {};
      \draw[->,line width=1] (r) -- (x) {}; \draw[->,line width=1] (l) -- (x) {};
      \draw[->,line width=1] (c) .. controls +(up:3.5) and +(up:3.5) .. (y) {};
    \end{tikzpicture}
    \end{center}
    \caption{
    The optimal number of traversals for the edges of this graph is non-unique. 
    The five successors of $x$ and the two possible paths to $x$ require that the algorithm traverses $(y,v_1,x)$ or $(y,v_2,x)$ multiple times. 
    }\label{unique-solution}
  \end{figure}
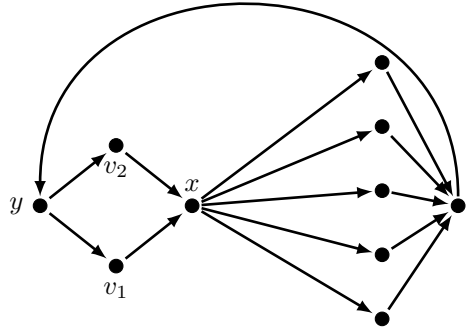
  
Our algorithm depends on the oracle and the fixed solution because it tries to reconstruct the fixed solution. 
Therefore, it ignores the cost function for the edges and tries to use the edges as often as their are used in the fixed optimal solution. 
Knowing, for every edge $e$, its membership to the three sets, the traversal number is unknown only for the edges from $\Ex$. 
Otherwise, the algorithm knows that it is used once or never. 
Because the edges from $\Ex$ need additional advice to compute their exact number of traversals, the oracle is interested in fixing a solution that minimizes the number of edges from $\Ex$. 
For the graph presented in \Cref{unique-solution}, the oracle would chose a solution where the edges $(y,v_2)$ and $v_2,x$ or  $(y,v_1)$ and $v_1,x$ are used exactly once. 
Such a solution sets the number of edges from $\Ex$ to three, instead of five. 

To ease the analysis of the advice complexity of our algorithms, the oracle constructs a graph $G'$ from the given one such that the optimal solution $S^*$ for $G'$ is unique in the sense that, for all other optimal solutions $S'$, $\#_{S^*}(e) = \#_{S'}(e)$. 
Additionally, every optimal solution for $G'$ must be an optimal solution for the input graph. 
The oracle constructs $G'$ by perturbing the costs on the edges by a small amount such that , for all $x$ and $y$, only one of two paths
$p_1=(y, \ldots, x)$ and $p_2=(y, \ldots, x)$ with equal costs is used multiple times.
This minimizes the number of edges from $\Ex$ and fixates the traversal number, for every edge. 
The cost function is perturbed by the following instruction: $\cost^\prime(e_i) = \cost(e_i) + 1/m^{2i}$. 
The following lemma guarantees that we can choose a sufficiently small perturbation such that $S^*$ is also optimal for $G$.

\begin{lemma}\label[lemma]{lem:trav-number}
  Let $S^*$ be an optimal exploration sequence for a graph with $n$
  vertices. Then $\#_{S^*}(e) \leq n$, for all $e\in E(S^*)$.
\end{lemma}
\begin{proof}
Assume that there exists an edge $e$ with $\#_{S^*}(e) > n$. 
This would imply that the algorithm sends the explorer more often along the edge $e$ than there are vertices in $G$. 
Thus, at least one of the traversal along $e$ is not needed to explore a vertex and this contradicts the optimality of $S^*$. 
Therefore, an additional cost of $1/m^{2i}$ for an edge $e_i$ will not sum up to $1$ which could change the optimal solution. 
\end{proof}

From now on, we assume that the optimal exploration sequence $S^*$ fixed by the oracle is unique. 

If the input graph has a solution that visits every vertex at least once, it must be strongly connected for the case of the cyclic graph exploration problem and it must be at least connected if the algorithm is asked for a non-cyclic exploration sequence. 
Because the connectivity of the graph alone does not reflect the current position of the explorer and the already exhausted edges, we introduce the term of an \emph{expandable search sequence} to formulate a stronger and more precise connectivity requirement. 

\begin{definition}\label{expandability}
We call a search sequence $S=(v_0, \ldots ,v_k)$ \emph{expandable} with respect to a fixed optimal exploration sequence $S^* = \{ v_0, \ldots, v_{end} \}$ if there exists an outgoing edge $e=(v_k,w)$ that is less used than in $S^*$ and there is a search sequence from $v_k$ that uses only not exhausted edges and reaches the final vertex $v_{end}$ such that $\cost(S^*)=\cost(S)$ and $V(S)=V$. 
\end{definition}
\section{Structural Observations}
\label{sec:basic}
In this section, we focus on the structural properties of the fixed optimal solution $S^*$. 
We distinguish between two different graphs, that can be induced by $S^*$. 
If we remove the edges from $\En$ and just look at the graph $G_{S^*}=(V,\Eu)$, we have a graph for which all edges are needed in the optimal exploration sequence. 
If we are more interested on traversed subsequences and want to distinguish different traversals over the same edge from $\Ex$, we look at the multigraph $M_{S^*}=(V,\Eu, \#_{S^*})$. 
In $M_{S^*}$, every edge $e=(v,u) \in \Ex$ is replaced by $\#_{S^*}(e)$ many edges that point from $v$ to $u$. 
In the case of the cyclic graph exploration problem, the multigraph $M_{S^*}=(V,\Eu, \#_{S^*})$ is an Eulerian graph. 
Let $d_{in}(v)$ denote the in-degree of the vertex $v$ in $M_{S^*}$.

\begin{definition}
Let $G=(V,E)$ be a graph with the optimal exploration sequence $S^*$, we denote for a vertex $v \in V$ the difference between the number of traversals for the outgoing edges and the incoming edges in $S^*$ with $d_\Delta(v) = \sum_{e=(v,x)}\#_{S^*}(e) - \sum_{e=(x,v)}\#_{S^*}(e)$. 
\end{definition}

Thus, for each vertex $v$ we have $d_\Delta(v)=0$ and $M_{S^*}$ is strongly connected. 
Note that $G_{S^*}$ is also strongly connected and $d_\Delta(v)=0$ holds. 
On the other hand, $M_{S^*}$ is an Eulerian path graph with start vertex $v_0$ and terminal vertex $v_{end}$ if $d_\Delta(v_0)=1$, $d_\Delta(v_{end})=-1$ and for each vertex $v \in V \setminus\{ v_0 , v_{end} \}$ we have $d_\Delta(v)=0$. 
Note that by adding the edge $(v_0,v_{end})$ to $\Eu$, $G_{S^*}$ becomes an Eulerian graph. 

We start by considering the cyclic graph exploration problems on directed graphs. 
The following lemmas describe structural properties of optimal exploration sequences which will be helpful for deducting the traversal numbers of the edges from a rather small amount of advice. 

\begin{lemma}\label[lemma]{zurueck}
  Let $G=(V,E)$ be a strongly connected graph, with $d_\Delta(v)=0$ for all vertices $v\in V$ and let $S^*$ be the optimal exploration sequence. 
  If there is a path from $y$ to $x$ that uses only edges $e$ with $\#_{S^*}(e)>0$, then there exists a cycle $(y, \ldots, x, \ldots, y)$ that uses only edges $e$ with $\#_{S^*}(e)>0$. 
\end{lemma}
\begin{proof}
	The graph $M_{S^*}=(V,\Eu, \#_{S^*})$ is Eulerian, in particular, it is strongly connected. 
	Assume that there is no path from $y$ to $x$ that uses only edges $e$ with $\#_{S^*}(e)>0$. 
  	Then we may partition $V$ into two sets $X$ and $Y$ with $X \cup Y = V$, $x \in X$ and $y \in Y$ such that there is no edge $e^\prime = (y^\prime, x^\prime) \in E$ with $\#_{S^*}(e^\prime)>0$, $y'\in Y$ and $x'\in X$. 
  	This contradicts the strong connectivity of $M_{S^*}$. 
\end{proof}

When the optimal solution contains multiple cycles that start at the same vertex, it is possible to interchange some of the cycles without changing the number of traversals for any edge and therefore, without changing the cost of the solution. 

\begin{lemma}\label[lemma]{interchangeable}
  Let $G=(V,E)$ be a connected graph, and let $S^*$ be the optimal exploration sequence for $G$. 
  A vertex $v$ is part of $d_{in}(v)$ cycles $c_i=(v, \ldots, v) \in S^*$ with $0 \leq i \leq d_{in}(v)$, and at least $d_{in}(v)-1$ cycles are interchangeable, i.e., changing their order in $S^*$ does not change the cost of $S^*$.
\end{lemma}
\begin{proof}
Changing the order of two cycles $c_i$ and $c_j$, which start at a vertex $v \in V \setminus \{ v_{0} \}$ and do not contain the starting vertex, does not change $\#_{S}(e)$ for any $e \in E(S^*)$. 
Note that all cycles that start at the starting vertex $v_0$ are interchangeable. 
Because, the cost of $S^*$ is defined as $\sum_{e\in E(S)}\cost(e)\cdot \#_{S}(e)$, the cost of $S^*$ is still minimal if $c_i$ and $c_j$ are exchanged. 
\end{proof}

Our goal is to prove that, due to the uniqueness of $S^*$, the edges from $\Ex$ cannot form a cycle, even in the underlying undirected graph. 
To this end, we first show that, in an optimal exploration sequence $S^*$, a single edge from $(v,w) \in \Ex$ cannot have a backward edge $(w,v)$ that is used in $S^*$. 

\begin{lemma}\label[lemma]{doppel}
   Let $G=(V,E)$ be a graph with the optimal exploration sequence $S^*$, then there is no vertex pair $v,w$ with $(v,w) \in \Ee$ and $(w,v) \in \Ex$. 
\end{lemma}
\begin{proof}
  Assume there is a pair $v,w$ with $(v,w) \in \Ee$ and $(w,v) \in \Ex$. 
  We define the following traversal function $c \colon E \to \mN_0$ with respect to the optimal solution $S^*$: 
  \[c(e) = \left\{\begin{array}{ll}
  \#_{S^*}(e)-1 & \mbox{for }e=(v,w)\mbox{ or }e=(w,v) \text{,}\\
  \#_{ S^*}(e) & \mbox{otherwise.}
  \end{array}\right.\]
  Thus, only one edge $(v,w)$ becomes unusable. 
  But the condition $d_\Delta(v)=0$ still holds for all vertices $v \in V$. 
  The edge $(w,v)$ is still usable, thus by applying Lemma \ref{zurueck} we get also a path from $v$ to $w$. 
  So $M=(V,\Eu,c)$ is also an Eulerian graph, thus there is an exploration sequence $S$ with less cost than $S^*$. 
  This is a contradiction to the optimality of $S^*$. 
\end{proof}

We can use the same contradicting argument as in \Cref{doppel} to prove that a directed cycle of multiply used edges cannot be part of an optimal solution. 
If we assume such a cycle exists in an optimal solution, it is possible to create a cheaper exploration sequence by reducing the number of traversals of the edges in the cycle by one.

\begin{lemma}\label[lemma]{kein-gerichteter-kreis}
  Let $G=(V,E)$ be a graph with the optimal exploration sequence $S^*$. 
  Let $S=(v_0, \ldots, v_{end})$ be a simple cyclic contiguous subsequence of $S^*$.
  Then there is an edge $e \in S(E)$ with $\#_{S^*}(e) = 1$.
\end{lemma}
\begin{proof}
  Let $S$ be a simple cyclic search sequence as seen in \Cref{bild-kreise} (a). 
  Assume that $\#_{S^*}(e) > 1$ holds for all edges in $E(S)$. 
  We now define a traversal function on $c \colon E \to \mN_0$ with respect to the optimal solution $S^*$: 
  \[c(e) = \left\{\begin{array}{ll}
  \#_{S^*}(e)-1 & \mbox{for }e\in E(S) \text{,}\\
  \#_{S^*}(e) & \mbox{otherwise.}
\end{array}\right.\]
  
  Note that $\#_{S^*}(e)-1>0$ for $e\in E(S)$, thus the multigraph $M=(V,E,c)$ is still connected. 
  Furthermore, the condition $d_\Delta(v)=0$ also holds for all vertices $v \in V$. 
  Thus, the graph $M=(V,E,c)$ is Eulerian and has an exploration sequence $S'$ with lower cost than the exploration sequence $S^*$. 
  This is a contradiction to the minimality of the cost of $S^*$.
\end{proof}

  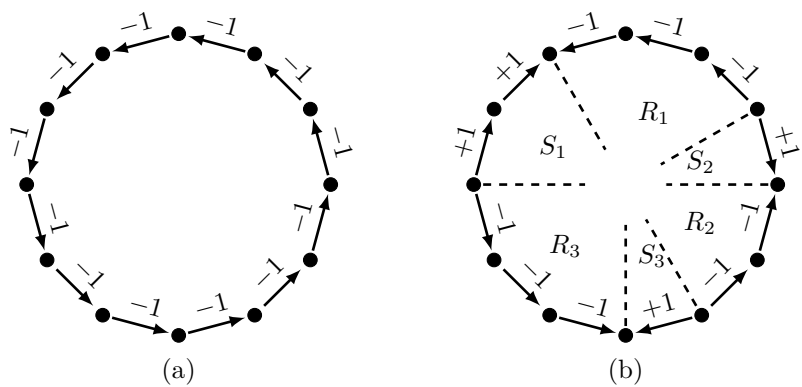
\begin{figure}[htb]
   \begin{center}
   \begin{tikzpicture}[>=latex]
     \foreach \x in {0,1,...,11} 
     {\draw (\x*30-30:2) node (\x) {}; \filldraw (\x) circle (\KnotenDiam);}
     \foreach \x in {0,1,...,11} {\TakeInt{A}{mod(\x+1,12)}\draw[->,line width=1] (\x) -- (\ResultA) node[above,sloped,midway] {$-1$};}
     
     \draw (0,-2) node[label={[label distance=0.05cm]270:(a)}] (0) {};
   \end{tikzpicture}\hskip10mm\begin{tikzpicture}[>=latex]
     \foreach \x in {0,1,...,11} {\draw (\x*30-30:2) node (\x) {};\filldraw (\x) circle (\KnotenDiam);}
     \foreach \x in {0,2,3,4,7,8,9,11} {\TakeInt{A}{mod(\x+1,12)}\draw[->,line width=1] (\x) -- (\ResultA) node[above,sloped,midway] {$-1$};}
     \foreach \x in {1,5,6,10} {\TakeInt{A}{mod(\x+1,12)}\draw[->,line width=1] (\ResultA) -- (\x) node[above,sloped,midway] {$+1$};}
     
      \draw[-,dashed,line width=1] (1) -- (0,0) node[above,sloped,midway] {}; 
      \draw (0,0) node[label={[label distance=0.55cm]3:$S_2$}] (blub) {};
     
      \draw[-,dashed,line width=1] (2) -- (0,0) node[above,sloped,midway] {};
      \draw (0,0) node[label={[label distance=0.55cm]87.9:$R_1$}] (blub) {};
      
      \draw[-,dashed,line width=1] (5) -- (0,0) node[above,sloped,midway] {};
      \draw (0,0) node[label={[label distance=0.55cm]160:$S_1$}] (blub) {};
      
      \draw[-,dashed,line width=1] (7) -- (0,0) node[above,sloped,midway] {};
      \draw (0,0) node[label={[label distance=0.55cm]230:$R_3$}] (blub) {};
      
      \draw[-,dashed,line width=1] (10) -- (0,0) node[above,sloped,midway] {};
      \draw (0,0) node[label={[label distance=0.55cm]273:$S_3$}] (blub) {};
      
      \draw[-,dashed,line width=1] (11) -- (0,0) node[above,sloped,midway] {};
      \draw (0,0) node[label={[label distance=0.55cm]340:$R_2$}] (blub) {};
     
     \filldraw (0,0) node[circle=3mm, draw, fill=white, color=white, minimum size=30pt] (blub) {}; 
     
     \draw (0,-2) node[label={[label distance=0.05cm]270:(b)}] (0) {};
   \end{tikzpicture}
   \end{center}
    \caption{Cycles in exploration sequences.}\label{bild-kreise}
  \end{figure}

Our next observation regarding the edges from $\Ex$ is that they cannot even form a cycle in the underlying undirected graph. 
For a given optimal exploration sequence $S^*$, we look at the underlying undirected graph and separate it into cycles. 
For each cycle, we summarize equally directed edges into sequences $S_i$ and $R_i$, as presented in \Cref{bild-kreise} (b). 
If we increment the traversal numbers for the sequences $S_i$ and decrease them for the sequences $R_i$, the resulting graph is still Eulerian. 
This leads to the following lemma.

\begin{lemma}\label[lemma]{kein-ungerichteter-kreis}
  Let $G=(V,E)$ be a graph with exploration sequence $S^*$ of minimal cost. 
  Let $S_i=(v^i,...,,u^i)$ and $R_i=(w^i,...,u^i)$, with $1\leq i\leq k$, be $2k$ simple search subsequences of $S^*$ with $w^i=v^{i+1}$ ($1\leq i< k$) and $v^1=w^k$. 
  Then there is an edge $e \in S_i \cup R_i$ with $\#_{S^*}(e) = 1$. 
\end{lemma}
\begin{proof}
  Assume that $\#_{S^*}(e) > 1$ holds for all edges $e \in \bigcup_{1\leq i\leq k}E(S_i)\cup E(R_i)$. 
  Without loss of generality, we assume that 
  \begin{equation}\label{eq1}
  \sum_{1\leq i\leq k}\cost(S_i) < \sum_{1\leq i\leq k}\cost(R_i) \text{.}
  \end{equation}
Again, a cheaper exploration sequence can be created by choosing the cheap sequences $S_i$ to a vertex $v$ more often than the expensive $R_i$. 
We define now a traversal function on $c \colon E\to \mN_0$ with respect to the optimal solution $S^*$: 
  \[c(e) = \left\{\begin{array}{ll}
  \#_{S^*}(e)+1 & \mbox{for }e\in \bigcup_{1\leq i\leq k}E(S_i) \text{ ,}\\
  \#_{S^*}(e)-1 & \mbox{for }e\in \bigcup_{1\leq i\leq k}E(R_i) \text{ ,}\\
  \#_{S^*}(e) & \mbox{otherwise.}\\
  \end{array}\right.\]
  Due to \eqref{eq1}, the newly created solution must be cheaper. 
  So, the multigraph $M=(V,E,c)$ is Eulerian and has an exploration sequence $S'$ with $\cost(S')< \cost(S^*)$. 
  This is a contradiction to the minimality of the cost of $S^*$. 
\end{proof}

We showed in \Cref{doppel,kein-gerichteter-kreis,kein-ungerichteter-kreis} that the information about the usage of edges in an unique optimal solution $S^*$ provides some structural information. 
Because the graph induced by $\Ex$ is cycle-free, it is a forest and we will from now on call it $F$.
This holds even for the underlying undirected graph. 
As next step, we show that it is possible to derive the exact traversal number for every edge from this induced structure. 

\begin{lemma}\label[lemma]{berechne-mehrfach}
  Let $G=(V,E)$ be a graph with exploration sequence $S^*=(v_0, \ldots,v_{end})$ of minimal cost. 
  There is an algorithm $\mathcal{A}(V,\En,\Ee,\Ex)$ which computes, for each edge $e\in E$, the number of visits by $S^*$.
\end{lemma}
\begin{proof}
  	Note that the algorithm $\mathcal A$ a priori only knows the exact number of visits for an edge $e$ if $\#_{S^*}(e)\leq 1$. 
  	\Cref{kein-gerichteter-kreis} and \Cref{kein-ungerichteter-kreis} show that the graph $F = (V,\Ex)$ is a forest, even if the edges are undirected. 
  	Furthermore, \Cref{doppel} shows that, if edges are used multiple times in one direction, they are not used at all in the reverse direction. 
  
 	 We now describe a recursive algorithm which computes the exact number of visits for the edges from $\Ex$. 
 	 A tree $T$ in the forest $F$ must have at least one leaf. 
 	 A leaf $v$ in $F$ has exactly one edge that is used multiple times and the other adjacent edges in $G$ are used exactly once or never. 
	 Let $U$ be the set of those edges from $\Ex$, where the number of traversals is unknown. 
	 As long as $U \not= \emptyset$, the graph $F^\prime = (V,U)$ has at least one leaf $v$. 
	 Let $e \in U$ be the only directed edge attached to $v$ in $F^\prime$. 
	 For all other edges $e^\prime \in E(G)$ attached to $v$, we already know the number of visits, i.e., $e^\prime \in \En \cup \Ee \cup \Ex \setminus U$. 
	 The algorithm $\mathcal A$ sums up the traversal numbers of all other incoming and outgoing edges for $v$, i.e., 
  \[\begin{array}{rcl}
  c_{in} & = & \sum_{(x,v)\in E \setminus e}\#_{S^*}(x,v) \text{,}\\
  c_{out} & = & \sum_{(v,x)\in E \setminus e}\#_{S^*}(v,x)\text{.}
  \end{array}\]
	Thus, the usage of the edge $e$ can be calculated, because $d_\Delta(v) = 0$ holds for $G_{S^*}=(V,\Eu)$. 
	Therefore, the number of traversals is the difference between incoming and outgoing edges from $\Ee$. 
	Only two cases need to be considered. 
	If $e=(w,v)$, then $c_{in} < c_{out}$ and $\mathcal A$ can compute its traversal number $\#_{S^*}(e)=c_{out}-c_{in}$.
	Analogously, for $e=(v,w)$, we have $\#_{S^*}(e)=c_{in}-c_{out}$. 
	
	The result can be used in the next level of the tree to calculate analogously the number of traversals for multiply used edges. 
	If the number of traversals for all edges $e \in \Ex$ attached to a leaf of $F^\prime = (V,U)$ are computed, we update $U' = U \setminus \{ e \}$. 
	Now, we have again new leaves in the tree $F^{\prime \prime}=(V,U')$ for which exactly one adjacent edge is part of $U'$. 
	For all other edges, the precise number of traversals is known because they are from $\En$, $\Ee$ or $\Ex \setminus U$. 
	
	Repeating this bottom-up approach for the tree $F^{\prime \prime}$ allows us to compute the exact usage of every edge from $\Ex$. 
\end{proof}
\section{Results for Known Graphs}
\label{sec:known}
If the structure of the graph is known to the algorithm, the existence of the cost function makes a difference. 
If there is a unit cost function, the algorithm can compute every feasible cyclic or non-cyclic exploration sequence beforehand and chooses the cheapest as solution. 

\begin{theorem}\label{bekannt-gerichtet}
  There exists an online algorithm which solves the cyclic and the non-cyclic graph exploration problem using no advice on known directed and known undirected graphs, if all edges have unit costs. 
\end{theorem}
\begin{proof}
  Because the algorithm is not limited in its computing power, it can compute all possible Hamiltonian cycles, or Hamiltonian tours, respectively, by brute force. 
  Obviously, this is possible for the directed and undirected case. 
  It chooses the solution which results in the shortest exploration sequences. 
\end{proof}

But, if the cost function is not known, the algorithm will not know which of the feasible exploration sequences will be the cheapest. 
The algorithm from \Cref{berechne-mehrfach} can be used to compute the optimal exploration sequence, if the optimal number of traversals for every edge and the structure of the graph is known. 
Thus, oracle and algorithm agree on a way to enumerate the edges and the oracle encodes for every single edge if it belongs to the set $\En$, $\Ee$, or $\Ex$. 
With this approach, we can solve the cyclic exploration problem for any directed graph with cost function. 

\begin{theorem}\label{bekannt-gerichtet-a}
  There exists an online algorithm which solves the cyclic graph exploration problem using $\lceil\log(3) m\rceil$ bits of advice on known directed graphs. 
\end{theorem}
\begin{proof}
The idea is that the edges from $E$ are split into three disjoint sets ($E=\En \cup \Ee \cup \Ex$) depending on the number of their visits during an optimal exploration sequence: 
The advice will indicate for each edge the corresponding class. 
Because the optimal solution $S^*$ is unique, \Cref{kein-gerichteter-kreis,kein-ungerichteter-kreis} are applicable. 

The online algorithm, which knows the directed graph $G$, proceeds as follows: 
\begin{enumerate}
  \item Compute the smallest $k$ with $2^k > 3^m$ and read $k$ bits from the advice tape. Interpret these $k$ bit as a binary number and transform it
    to a ternary number $a$ with $m$ digits. 
  \item Use these $m$ digits to classify all edges from $E$ into the classes $\En$, $\Ee$ and $\Ex$. 
  \item Use the algorithm from \Cref{berechne-mehrfach} to compute the values $\#_{S^*}(e)$ for the edges from $\Ex$. 
  \item Knowing now all values $\#_{S^*}(e)$ for all edges $e\in E$, it is possible to compute the optimal exploration sequence (the given Eulerian tour).\qedhere \end{enumerate} 
\end{proof}

If the algorithm has to solve the non-cyclic graph exploration problem with a cost function for the edges, it needs additional information about where the exploration sequence ends. 
Thus, the oracle encodes the identifier of the last vertex that is visited in the optimal exploration sequence, which increases the number of advice bits by $\log(n)$. 

\begin{theorem}\label{bekannt-gerichtet-b}
  There exists an online algorithm which solves the non-cyclic graph exploration problem using $\lceil\log(3) m\rceil+\lceil\log n\rceil$ bits of advice on known directed graphs.
\end{theorem}
\begin{proof}
  First, the algorithm reads $\lceil\log n\rceil$ bits to know the final vertex $v_{end}$ in the optimal exploration sequence $S^*$. 
  Note that the algorithm now knows that $d_{in}(v_{end}) = d_{out}(v_{end}) -1$. 
  Then it reads the same advice as in \Cref{bekannt-gerichtet} and uses the method from the proof of \Cref{berechne-mehrfach} to compute also the values $\#_{S^*}(e)$ for the edges from $\Ex$.
\end{proof}

Analogous proofs can be given for the undirected variations of the problems from \Cref{bekannt-gerichtet-a,bekannt-gerichtet-b}.
The algorithm virtually separates every edge into two directed edges. 
These virtual directed edges can be used once, multiple times or never in an optimal exploration sequence. 
From \Cref{doppel}, we know that, if one virtual directed edge is used multiple times, the other one cannot be used. 
Therefore, an undirected edge can be classified into six cases to denote the number of traversals for the virtual directed edges. 
Either both virtual directed edges are never used or used once or exactly one directed edge is used once or multiple times. 

\begin{theorem}\label{bekannt-ungerichtet}
  There exists an online algorithm which solves the cyclic graph exploration problem using $\lceil\log(6) m\rceil$ bits of advice on known undirected graphs. 
  For the non-cyclic graph exploration problem, the algorithm needs $\lceil\log(6) m\rceil+\lceil\log n\rceil$ bits of advice. 
\end{theorem}
\begin{proof}
  When the given graph $G=(V,E)$ is undirected, the online algorithm does not know in which direction an edge $\{v,w\}$ is used. 
  Using \Cref{doppel}, the following traversals are possible for any edge $\{v,w\}$:
  \begin{enumerate}
  \item $\#_{S^*}(v,w) = 0$ and $\#_{S^*}(w,v) = 0$, 
  \item $\#_{S^*}(v,w) = 1$ and $\#_{S^*}(w,v) = 0$, 
  \item $\#_{S^*}(v,w) = 0$ and $\#_{S^*}(w,v) = 1$, 
  \item $\#_{S^*}(v,w) = 1$ and $\#_{S^*}(w,v) = 1$, 
  \item $\#_{S^*}(v,w) > 1$ and $\#_{S^*}(w,v) = 0$, 
  \item $\#_{S^*}(v,w) = 0$ and $\#_{S^*}(w,v) > 1$. 
  \end{enumerate}
  Thus, the advice which is necessary to have the same information as in the proof of \Cref{bekannt-gerichtet-a} (\Cref{bekannt-gerichtet-b}, resp.), is at most $\lceil\log(6)\cdot m\rceil$ ($\lceil\log(6)\cdot m\rceil+\lceil\log n\rceil$, resp.).
  Therefore, we can proceed in the same way as before.
\end{proof}

So, if the graph is known to the algorithm, we proved for every presented variation of the problem that a linear number of advice bits suffices to compute an optimal exploration sequence. 
\section{Unknown Directed Graphs with In- and Outdegree Two}
\label{sec:bounded}
We cannot apply the previous approach for the setting of unknown graphs because the lack of edge information prohibits the edge classification at the beginning. 
Even if the explorer asks, for every newly revealed edge, in which set it is contained, \Cref{berechne-mehrfach} can only be applied after all edges are explored. 
Therefore, we now informally explain an extended algorithm that requires that every vertex in the input graph has at most two incoming and also at most two outgoing edges. 
In \Cref{sec:general}, we will show how every graph can be transformed into such a degree-bounded graph. 
We now give a high level description of an algorithm with advice that computes the precise number of traversals while traversing such a given degree-bounded graph and give a detailed analysis together with its correctness in the next subsection. 

As soon as the explorer resides in a vertex $v$, the algorithm gets the identifiers of $v$ and of all out-neighbors $N_{out}(v)$ of $v$. 
If a neighbor $w \in N_{out}(v)$ was already the out-neighbor of a previously visited vertex, the algorithm recognizes this vertex. 
The vertices $N_{in}(v)$ that lead to $v$ and their corresponding edges stay hidden if the explorer is positioned at $v$. 
So, visiting a vertex $v$ reveals only its successors $N_{out}(v)$ and not its predecessors $N_{in}(v)$. 
The algorithm accesses the advice tape to overcome this lack of information. 

On the first visit of $v$, the algorithm reads advice bits to classify the outgoing edges into the three sets $\En$, $\Ee$, or $\Ex$ and to know the number of incoming edges adjacent to $v$ and their membership in the sets $\En$, $\Ee$, or $\Ex$. 
So, the algorithm knows, for every visited vertex, all incoming and outgoing edges, and whether they are used once, more than once, or never in an optimal exploration sequence. 
The edges from the set $\En$ are ignored in the further process, so they are not involved in computing the exploration sequence. 
Obviously, the graph $G_{S^*}=(V,\Eu)$ is also bounded in the number of adjacent vertices if the input graph is bounded. 
Therefore, every vertex in $G_{S^*}$ looks like one of the vertices shown in \Cref{possible-vertices}. 

\begin{figure}[htb]
    \begin{center}
    \begin{tikzpicture}[scale=0.8, >=latex]
      \draw (0,0) node[label={[label distance=0.05cm]120:$v_1$}] (v1) {};\filldraw (v1) circle (\KnotenDiam);
      \draw[->,line width=1.5] (v1) -- (1,0) node[above,sloped,midway] {};
      \draw[<-,line width=1.5] (v1) -- (-1,0) node[above,sloped,midway] {};
      
      \draw (3,0) node[label={[label distance=0.05cm]120:$v_2$}] (v2) {};\filldraw (v2) circle (\KnotenDiam);
      \draw[->,line width=1.5] (v2) -- (4,1) node[above,sloped,midway] {};
      \draw[->,line width=1.5] (v2) -- (4,-1) node[above,sloped,midway] {};
      \draw[<-,line width=1.5] (v2) -- (2,0) node[above,sloped,midway] {};
      
      \draw (6,0) node[label={[label distance=0.05cm]180:$v_3$}] (v3) {};\filldraw (v3) circle (\KnotenDiam);
      \draw[->,line width=1.5] (v3) -- (7,0) node[above,sloped,midway] {};
      \draw[<-,line width=1.5] (v3) -- (5,1) node[above,sloped,midway] {};
      \draw[<-,line width=1.5] (v3) -- (5,-1) node[above,sloped,midway] {};
      
      \draw (9,0) node[label={[label distance=0.05cm]180:$v_4$}] (v4) {};\filldraw (v4) circle (\KnotenDiam);
      \draw[->,line width=1.5] (v4) -- (10,1) node[above,sloped,midway] {};
      \draw[->,line width=1.5] (v4) -- (10,-1) node[above,sloped,midway] {};
      \draw[<-,line width=1.5] (v4) -- (8,1) node[above,sloped,midway] {};
      \draw[<-,line width=1.5] (v4) -- (8,-1) node[above,sloped,midway] {};
      
      \draw (12,0) node[label={[label distance=0.05cm]180:$v_5$}] (v5) {};\filldraw (v5) circle (\KnotenDiam);
      \draw[->,line width=1.5] (v5) -- (13,-1) node[above,sloped,midway] {};
      \draw[->,gray,line width=1.0] (v5) -- (13,1) node[above,sloped,midway] {};
      \draw[<-,line width=1.5] (v5) -- (11,1) node[above,sloped,midway] {};
      \draw[<-,line width=1.5] (v5) -- (11,-1) node[above,sloped,midway] {};
      
      \draw (15,0) node[label={[label distance=0.05cm]180:$v_6$}] (v6) {};\filldraw (v6) circle (\KnotenDiam);
      \draw[->,line width=1.5] (v6) -- (16,-1) node[above,sloped,midway] {};
      \draw[->,line width=1.5] (v6) -- (16,1) node[above,sloped,midway] {};
      \draw[<-,gray,line width=1.0] (v6) -- (14,1) node[above,sloped,midway] {};
      \draw[<-,line width=1.5] (v6) -- (14,-1) node[above,sloped,midway] {};
      
      \draw (0,-3) node[label={[label distance=0.05cm]180:$v_7$}] (v7) {};\filldraw (v7) circle (\KnotenDiam);
      \draw[->,line width=1.5] (v7) -- (1,-4) node[above,sloped,midway] {};
      \draw[->,gray,line width=1.0] (v7) -- (1,-2) node[above,sloped,midway] {};
      \draw[<-,gray,line width=1.0] (v7) -- (-1,-2) node[above,sloped,midway] {};
      \draw[<-,line width=1.5] (v7) -- (-1,-4) node[above,sloped,midway] {};
      
      \draw (3,-3) node[label={[label distance=0.05cm]180:$v_8$}] (v8) {};\filldraw (v8) circle (\KnotenDiam);
      \draw[->,line width=1.5] (v8) -- (4,-3) node[above,sloped,midway] {};
      \draw[<-,gray,line width=1.0] (v8) -- (2,-2) node[above,sloped,midway] {};
      \draw[<-,line width=1.5] (v8) -- (2,-4) node[above,sloped,midway] {};
      
      \draw (6,-3) node[label={[label distance=0.05cm]120:$v_{9}$}] (v9) {};\filldraw (v9) circle (\KnotenDiam);
      \draw[->,line width=1.5] (v9) -- (7,-4) node[above,sloped,midway] {};
      \draw[->,gray,line width=1.0] (v9) -- (7,-2) node[above,sloped,midway] {};
      \draw[<-,line width=1.5] (v9) -- (5,-3) node[above,sloped,midway] {};
      
      \draw (9,-3) node[label={[label distance=0.05cm]120:$v_{10}$}] (v10) {};\filldraw (v10) circle (\KnotenDiam);
      \draw[->,gray,line width=1] (v10) -- (10,-4) node[above,sloped,midway] {};
      \draw[->,gray,line width=1.0] (v10) -- (10,-2) node[above,sloped,midway] {};
      \draw[<-,line width=1.5] (v10) -- (8,-3) node[above,sloped,midway] {};
      
      \draw (12,-3) node[label={[label distance=0.05cm]180:$v_{11}$}] (v11) {};\filldraw (v11) circle (\KnotenDiam);
      \draw[->,line width=1.5] (v11) -- (13,-3) node[above,sloped,midway] {};
      \draw[<-,gray,line width=1] (v11) -- (11,-2) node[above,sloped,midway] {};
      \draw[<-,gray,line width=1] (v11) -- (11,-4) node[above,sloped,midway] {};
     
      \draw (15,-3) node[label={[label distance=0.05cm]120:$v_{12}$}] (v12) {};\filldraw (v12) circle (\KnotenDiam);
      \draw[->,gray,line width=1] (v12) -- (16,-3) node[above,sloped,midway] {};
      \draw[<-,gray,line width=1] (v12) -- (14,-3) node[above,sloped,midway] {};
    \end{tikzpicture}
    \end{center}
    \caption{All possible edge configurations for $G_{S^*}$ with bounded degree. 
    The thick black edges are from $\Ex$. 
    The gray edges are from $\Ee$. }\label{possible-vertices}
\end{figure}
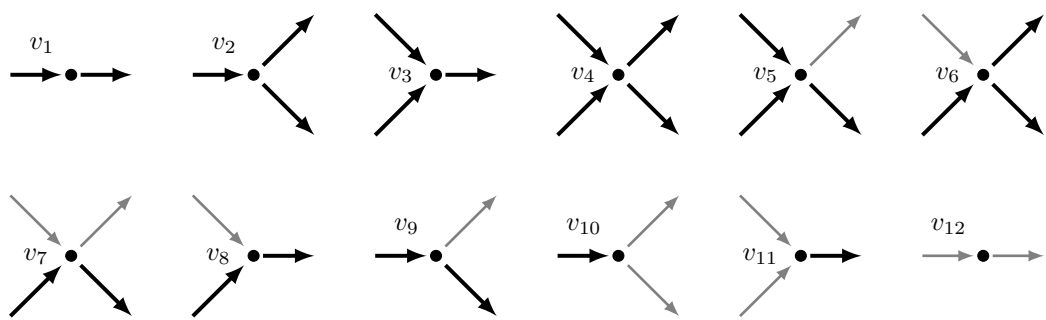

Our algorithm will prefer a path that is relatively often used and will avoid to exhaust an edge. 
As long as we know that the explorer does not use an edge for the last time, we know that the incident vertices are visited again. 
Thus, the algorithm chose one of many exchangeable cycles and did not make any mistake. 
If the explorer stands on a vertex with only one outgoing edge, the algorithm does not need to make a decision because the only valid move is to use this edge. 

Therefore, we assume that the explorer is positioned at a vertex where the algorithm has to decide between two edges which were not used in some previous step. 
If there is only one edge of $\Ex$, the algorithm will prefer it without knowing its number of traversals. 
But, if both outgoing edges are used multiple times, the lack of precise traversal numbers makes it impossible for the algorithm to choose the more often used edge. 
Thus, the algorithm accesses the advice tape to compare which edge is more often used in an optimal exploration sequence. 

\begin{definition}\label{heavy-light}
Let $G$ be a graph where every vertex has at most two outgoing and two incoming edges and let $S^*$ be an exploration sequence of minimal cost. 
For a vertex $v$ with two outgoing edges $e_1$ and $e_2$ from $\Ex$ we call $e_1$ \emph{light} if $\#_{S^*}(e_1) < \#_{S^*}(e_2)$. 
The other edge $e_2$ will be called \emph{heavy}. 
If $v$ has only one outgoing edge from $\Ex$, it is also called heavy. 
Analogously, we define for one or two incoming multi-edges of $v$ one edge as heavy and, if a second edge exists, the fewer traversed edge as light. 
\end{definition}

If both edges are equally often used, the oracle can decide arbitrarily which edge is heavy. 
The algorithm will ask for the exact number of traversals for the light edges and moves the explorer along the heavy one. 

Thus, the algorithm moves the explorer along a path of edges for which the precise number of traversals is unknown, but larger than one. 
But, the algorithm will know the exact number of traversals for all the {\em light} edges (which are not yet used for leaving the vertex for the first time). 
From \Cref{kein-ungerichteter-kreis}, we know that the path of multi-edges, which are preferred by the algorithm, are part of a tree $T$ in the forest $F=(V, \Ex)$. 
When the explorer reaches a leaf $w \in T$, the algorithm knows how often $w$ will be visited, because $w$ has only outgoing edges from $\Ee$, like $v_5$ in \Cref{possible-vertices}. 

Additionally, due to the bounded degree and the Eulerian property, $w$ can have only one incoming edge $e \in \Ex$ and therefore, the algorithm knows the exact number of outgoing traversals for $e$ is exactly $2$. 
This number of traversals can be used to compute the number of traversals for preceding edges, similar to \Cref{berechne-mehrfach}. 
For every preceding vertex with two outgoing edges, there was at most one incoming and one outgoing edge for which the number of traversals was unknown. 
Since, the number of traversals for the outgoing heavy edge can be computed, there is only one incident edge with an unknown number of traversals left. 
Thus, the algorithm can compute the number of traversals for all heavy edges adjacent to the traversed path in a bottom-up approach as soon as the explorer reaches a leaf $w \in T$. 
So, when a vertex is visited for the second time, the algorithm knows, for all adjacent edges, their precise number of traversals in the optimal solution.

We would like to summarize the idea briefly: 
If the explorer is on a vertex that is part of a tree of multi-edges for the first time, the algorithm asks for both, incoming and outgoing edges which are light, their precise number of traversals, and follows the outgoing heavy edge. 
At the point where the explorer has to leave the tree of multi-edges, the algorithm can compute the number of traversals for the heavy edges adjacent to the traversed path. 
Thus, the algorithm moved the explorer only once along the heavy multi-edges, which are more often used than once, to compute their precise number of traversals. 
So, when a vertex is visited for the second time, the algorithm knows, for all adjacent edges, their precise number of traversals in an optimal solution. 

Often, the algorithm just chooses between some interchangeable cycles, but it can happen that the graph has some critical edges whose early traversal would prevent the algorithm from computing an optimal exploration sequence. 
If the algorithm is in the situation to choose between two edges with only one traversal remaining (or two edges from $\Ee$), it asks for advice to know which edge should be used as {\em last} edge to leave the vertex. 
We say that such an edge is important to sustain the expandability (see \Cref{expandability}) of the current exploration sequence. 

\subsection{Correctness of Expanding the Search Sequence} 
In this subsection, we explain the algorithm in detail and prove its correctness. 
Recall that we assume $G_{S^*}$ has bounded in- and out-degree. 
An exploration sequence $S=(v_0,...,v_s)$ contains a vertex $v$ as often as its outgoing edges $e_1$ and $e_2$ are traversed. 
This implies that $v$ is part of $\#_{S}(e_1) + \#_{S}(e_2) $ cycles of which all but one can be traversed in any arbitrary order, as shown in \Cref{interchangeable}. 
Thus, as long as the {\em last} edge is not exhausted before the other one, the algorithm constructs an optimal solution. 
Note that because the graph is not known, the number of traversals in the optimal exploration sequence $\#_{S^*}(e)$ and the marking which edge is {\em last} are not known to the algorithm. 
This lack of information will be resolved with advice. 
Thus, we assume that the algorithm knows, for each vertex, which edge is the {\em last} one that is used in the optimal solution $S^*$, and for every edge its membership in the sets  $\En$, $\Ee$, and $\Ex$. 
On a second visit of a vertex, the algorithm also knows the precise number of traversals for all adjacent edges. 
This knowledge is necessary to recognize when the remaining number of traversals for an edge becomes critical, e.g., when only one traversal remains. 
We will explain in the next subsection how many bits are necessary to obtain this knowledge and how the algorithm will use the advice bits. 

\begin{theorem}\label{sicher-weg}
  Let $G=(V,E)$ be a graph with the optimal exploration sequence $S^*=(v_0,\ldots,v_{end})$ and let $S=(v_0, \ldots ,v_p)$ be an expandable prefix of it such that, for all $v\in V(S^*)$ the algorithm knows the edge through which $v$ is left for the last time in $S^*$, and, for all $e\in E(S^*)$, it knows whether $\#_{S^*}(e) - \#_{S}(e)>1$ or $\#_{S^*}(e) - \#_{S}(e)$ or $\#_{S^*}(e) - \#_{S}(e) = 0$.
  Then $S$ can be extended to an exploration sequence $S'=(v_0,v_p,v_{p+1}, \ldots, v_{end})$ with $\cost(S^*)=\cost(S')$. 
\end{theorem}
\begin{proof}
For each response $p$, we define a function $c_p$ that computes, for every edge $e$, the difference between the number of already used traversals in $S$ and the number of traversals in $S^*$: $c_p(e)=\#_{S^*}(e)-\#_{S}(e)$. 
The algorithm chooses to move the explorer along one of the adjacent edges, depending on the remaining traversals for the possible edges, represented by the function $c_p$. 
We write $c$ if $p$ is clear. 

The algorithm will never choose an exhausted edge to extend the sequence. 
The algorithm distinguishes between the following cases to extend its search sequence $S=(v_0,..., v_p)$:
  \begin{enumerate}
  \renewcommand{\labelenumi}{(\alph{enumi})}
  \item
  	If there is an edge $e=(v_p,w)$ with more than one remaining traversal $c(e)>1$, the algorithm chooses this edge to move the explorer to $w$. 
  	This extends the initial sequence to $S=(v_0,\ldots, v_p,w)$. 
    	The expandability of $S$ is preserved by this choice. 
    	Thus the explorer will visit $v_p$ again and choose again one of the outgoing edges, but this time $c(e)$ is reduced by one. 
    	This case repeats until $c(e)\leq 1$, for all outgoing edges of $v_p$. 
  \item
  	If there is an edge $e'=(v_p,w')$ with $c(e')=1$ which is not marked as {\em last}, the algorithm chooses it to extends the initial sequence to $S=(v_0, \ldots, v_p, w)$. 
    	This means that all search sequences from $v_p$ to $v_{end}$ have to traverse the edge that is not marked as {\em last} before the {\em last} edge is chosen.
    	Thus, the {\em last} edge $e=(v_p,w'')$ is the only remaining edge. 
  \item
	If there is only one edge $e=(v_p,w)$ with $c(e)=1$, it is the only adjacent edge of $v_p$ or it is marked as {\em last}. 
	Therefore , the algorithm moves the explorer to $w$ and all adjacent edges of $v_p$ are exhausted. 
	Thus, the initial sequence is extended to $S=(v_0, \ldots, v_p, w)$ and the vertex $v_p$ will not be visited again. 
  \end{enumerate}
  
We prove that this approach always finds an exploration sequence $S^\prime$ with $\cost(S^*)=\cost(S^\prime)$  by contradiction. 
Because the algorithm will never traverse an edge more often than it is used in the optimal solution, the costs will never exceed the costs of the optimal solution. 
We assume for the sake of contradiction that the computed sequence does not visit every vertex. 
Thus, there exists at least one vertex $u$ that is not visited. 
We look at the vertex $v$ that has an edge $e = (v,u) \in \Eu$. 
Such a vertex $v$ exists since $S^*$ visits every vertex, in particular $u$. 
So, the algorithm missed the edge $e$.

If $e \in \Ex$, the explorer traverses it on the first arrival on $v$ unless there exists another edge from $\Ex$ that is heavy. 
But if $e$ is light, the explorer traverses exchangeable cycles until both outgoing edges have only one remaining traversal. 
Therefore, the algorithm has to use $e$ if it is from $\Ex$. 

On the other hand, if $e \in \Ee$, it can be marked as {\em last} or not. 
If it is not marked as last, it is either the only outgoing edge of $v$, which makes it impossible to not visit $u$, or there exists another edge that is marked as {\em last}. 
But then, the algorithm will send the explorer along $e$ before the {\em last} edge is exhausted. 
If $e$ is marked as {\em last} edge, the explorer will always come back to $v$ until it uses $e$, due to the Eulerian property of $M_{S^*}$. 
Therefore, the algorithm will not leave out a vertex during its computation. 
\end{proof}
\subsection{Upper Bound for the Advice Complexity}

To overcome the lack of information that arises from the online scenario, the algorithm needs advice. 
In \Cref{sec:known} we showed that it was crucial for the algorithm to know, for every edge, its number of traversals in an optimal solution. 
The first step to compute the number of traversals for every edge is the classification of an edge into one of the three sets $\En$, $\Ee$, or $\Ex$. 
Because the algorithm sees only the outgoing edges of visited vertices, additional advice bits are necessary to learn the number of incoming edges. 

\begin{lemma}\label[lemma]{incoming-edges}
The algorithm can learn, for every edge, its membership to the sets $\En$, $\Ee$, or $\Ex$ on the first visit of an incident vertex using $n+\log(3)m$ advice bits overall. 
\end{lemma}
\begin{proof}
Assume the algorithm moves the explorer onto the vertex $v$ for the first time (or the starting position is $v$). 
All successors $N_{out}(v)$ of $v$ are visible together with their edges. 
The algorithm will read $\log(3)$ bits of advice for each of these outgoing edges. 
This classifies all outgoing edges of the vertex into one of the three sets. 

The in-degree of the vertices is unknown to the algorithm. 
But due to the bounded degree and the strong connectivity, it is either one or two. 
Thus, the algorithm reads on the first visit of a vertex one bit of advice to know the in-degree. 
If $v$ is not the starting vertex, the algorithm already used one of the incoming edges. 
Thus, it was already classified and it can read $\log(3)$ bits to know the membership of the other not yet used edge, if it exists. 
If $v$ is the starting vertex, the oracle writes the advice bits for classification in the same order in which the incident vertices will be explored. 
Whenever the explorer is moved for the first time onto some predecessor $u_i$ of $v$, it knows that the edge $(u_i,v)$ was classified as the first incoming edge of $v$. 
And if the explorer reaches the second predecessor, the edge leading to $v$ is the one that was classified as the second incoming edge of $v$. 

Thus, for every vertex, the algorithm reads if the number of incoming edges is one or two and, for every edge, $\log(3)$ bits to determine its membership to the sets $\En$, $\Ee$, or $\Ex$.
\end{proof}
All edges from the set $\En$ are not considered for any further decision of the algorithm and are subsequently ignored. 
The next important information given by the oracle is a comparison between the number of traversals for the equally directed multi-edges. 
The algorithm needs to know which of two possible outgoing multi-edges is light, recall that a light edge is less often used in the optimal solution $S^*$ than the other outgoing edge, according to \Cref{heavy-light}. 

\begin{lemma}\label{light-last}
For a given graph $G_{S^*}=(V,\Eu)$ with in- and out-degree bounded by $2$, with an optimal exploration sequence $S^*$, the algorithm needs at most $3n$ bits of advice to learn about the light incoming and outgoing edges and which outgoing edge is \emph{last}, for each vertex. 
\end{lemma}
\begin{proof}
For two outgoing edges $e_1$ and $e_2$ which are used multiple times in an optimal solution $S^*$, the algorithm asks if $\#_{S^*}(e_1) < \#_{S^*}(e_2)$ holds. 
Or in other words, the algorithm asks if the edge $e_1$ is light. 
This can be determined with one bit of advice. 
If a vertex has two incoming edges from $\Ex$ the algorithm will also ask for one bit of advice to know which incoming edge is light. 
Additionally, the algorithm asks for two outgoing edges which of them should be used for the {\em last} time to leave the vertex if both of them are in danger of being exhausted. 
The algorithm will ask these three binary questions at most once per vertex in the graph. 
\end{proof}

The next step of the algorithm is to ask, for every light edge, its exact number of traversals in the optimal solution $S^*$. 
To determine the cost for this information, we need to make some observations about the possible number of traversals of an edge and how this relates to the traversals of the incident vertices. 
A vertex $v$ is as often visited as the sum of traversals of its outgoing or incoming edges due to the Eulerian property of the graph $M_{S^*}$. 

\begin{definition}
The number of traversals for a vertex $v$ with respect to an optimal exploration sequence $S^*$ for an Eulerian graph $H=(V',E')$ is defined as $d_{S^*}(v) = \sum_{e=(v,w_i)\in E' } \#_{S^*}(e)$.
\end{definition}

We already mentioned above (see \Cref{kein-ungerichteter-kreis}) that every cycle contains at least one edge that is used only once. 
This implies that the number of visits for a vertex is bounded. 

\begin{lemma}\label{max-degree}
For every vertex $v$ in an Eulerian graph $G=(V,E)$ with optimal exploration sequence $S^*$, the number of traversals is $d_{S^*}(v) \leq n$. 
\end{lemma}
\begin{proof}
The explorer will visit the vertex $v$ exactly $d_{S^*}(v)$ times. 
Thus, there are at most $d_{S^*}(v)$ many cycles that contain $v$. 
In every cycle, there must be at least one vertex that is visited for the first time. 
Thus, $d_{S^*}(v) \leq n$ holds for every vertex $v \in V$. 
\end{proof}

We recall \Cref{kein-ungerichteter-kreis} and remember that a graph induced by the edges from $\Ex$ is a forest $F=(V,\Ex)$. 
\Cref{max-degree} bounds the number of different cycles that one vertex in one tree $T \in F$ is contained in. 
But, this lemma can be extended to bound the number of different cycles for all trees $T \in F$. 

\begin{lemma}\label{max-cycles} 
Let $M_{S^*}=(V,\Eu, \#_{S^*})$ be an Eulerian multigraph with an optimal exploration sequence $S^*$. 
Every edge $e \in \Ex$ is replaced by $\#_{S^*}(e)$ parallel edges. 
The multigraph $M_{S^*}$ can be decomposed into at most $n$ edge-disjoint cycles. 
\end{lemma}
\begin{proof}
Let there be $k$ different trees $T_1, \ldots, T_k$ of multi-edges. 
We choose the vertex with the largest number of outgoing traversals in $T_i$ as a representative vertex for this tree and denote it by $v_i$ and its number of traversals with $d_{S^*}(v_i)$. 
Every $v_i$ is part of $d_{S^*}(v_i)$ different cycles which have to explore at least one new vertex each. 
These cycles can be used to decompose the multigraph $M_{S^*}$ into edge-disjoint cycles. 

We prove the claim by contradiction and assume that $M_{S^*}$ is decomposed into more than $n$ edge-disjoint cycles. 
This would mean that there is at least one cycle that does not contain a vertex that is visited for the first time. 
If a cycle would not discover a new vertex, we could remove the cycle which would lead to a better solution or destroy the connectivity of the graph. 
Because $S^*$ is optimal, it is not possible that removing this cycle improves the solution. 

If removing the cycle destroys the connectivity, the cycle was necessary to explore a cut-off component. 
But then, there must be at least one vertex in this component that is discovered for the first time. 
Therefore, the assumption that the graph can be decomposed into more than $n$ cycles is wrong and the claim is proven. 
\end{proof}

With the knowledge that, for an optimal exploration sequence, the number of cycles in $G_{S^*}$ connecting different trees from $F$ is bounded by the number of vertices, we are able to prove that, if we choose one representative vertex $v_i$ for every tree $T_i$ the sum of their traversal numbers is bounded by the number of edges. 

\begin{lemma}\label{sum-traversals}
Let $G_{S^*}=(V,\Eu)$ be a graph induced by the edges of an optimal exploration sequence $S^*$. 
Let $T_1, \ldots, T_k$ be the trees induced by $\Eu$. 
Let $v_i \in V(T_i)$ be a representative vertex in $T_i$. 

Then, 
\begin{align*}
\sum_{i=1}^{k}d_{S^*}(v_i) \leq m \text{.}
\end{align*}
\end{lemma}
\begin{proof}
Every vertex $v_i$ is part of $d_{S^*}(v_i)$ cycles. 
If we interpret every multi-edge as a set of parallel edges that are used once, we know from \Cref{max-cycles} that these cycles are edge-disjoint. 
The edge-disjoint cycles can either contain only one $v_i$ or more of the representative vertices which would mean that these share an edge-disjoint cycle. 
To exit any $T_i$, we always use an edge from $\Ee$. 
Thus, if a cycle contains only one $v_i$, we need at least one edge from $\Ee$ to close this cycle. 
Therefore, every cycle that traverses only one tree $T_i$ has at least one edge that is used precisely once. 
Thus, the number of cycles that contain only one $v_i$ is less or equal to the number of edges from $\Ee$. 

If a cycle contains more than one of the representative vertices, there are at least as many edges from $\Ee$ in the cycle as representative vertices, because they were chosen from different trees $T_i$. 
The cycles that traverse multiple trees use at least one edge from $\Ee$ for every step from one tree $T_i$ to the next one $T_j$. 
Thus, a cycle has as many edges from $\Ee$ as it visits different trees $T_i$. 
So, we can allocate an edge from $\Ee$ to every representative vertex $v_i$. 
Thus, the sum of traversals over all these vertices is bounded by the number of edges from $\Ee$. 
\[ \sum_{i=1}^{k} d_{S^*}(v_i) \leq |\Ee| \leq m \text{.}  \qedhere \] 
\end{proof}

So, for any set of representative vertices chosen from different trees $T_i$, the sum of their traversal numbers is bounded. 
This gives a first intuition that the advice for the precise number of traversals for the light edges could be bounded. 
But, we need to discuss how the number of traversals $d_{S^*}(v)$ of a vertex is transferred to its outgoing edges $e_1$ and $e_2$. 
Recall that the algorithm only asks for advice if it has to choose between two edges from $\Ex$. 
Note that the following part concentrates on the outgoing edges and we will consider the symmetric case of incoming edges later. 

So, the algorithm asks if $e_1$ is a light edge, i.e., whether $\#_{S}(e_1) \leq \#_{S}(e_2)$. 
Because the number of traversals for one vertex is bounded by $n$ (see \Cref{max-degree}), the less used outgoing edge is used at most $\frac{n}{2}$ times. 
The algorithm asks the oracle for the exact number of traversals for the light edge. 
Because we do not want the algorithm to read $\log(\frac{n}{2})$ advice bits every time it asks for the number of traversals, we use self-delimiting codes (see \cite{Komm16}). 
Therefore, we use $\log(x)+2\log \log(x) + 2$ bits of advice to read the natural number $x$. 
As described in \Cref{sec:bounded}, the algorithm sends the explorer along the edges with a high number of traversals, e.g., the heavy edges in the first traversal of a vertex. 
If $v$ has at most one outgoing multi-edge, the algorithm will not use advice bits and mark it immediately as heavy. 
This strategy leads to a traversed path of multi-edges for which we do not know the number of traversals. 
Note that, for every edge from $\Ee$, the exact number of traversals is one and known. 
If a vertex $v$ has only outgoing edges that are used exactly once, the explorer knows immediately the number of traversals $d_{S^*}(v)$ for this vertex. 

\begin{figure}[htb]
    \begin{center}
    \begin{tikzpicture}[>=latex]
      \draw (0,0) node[label={[label distance=0.09cm]180:$v$}] (v) {};\filldraw (v) circle (\KnotenDiam);
      \draw (8,0) node[label={[label distance=0.09cm]180:$u$}] (u) {};
      
      \foreach \x/\y/\name in {1/1/l1, 1/-1/r1, 2/0/l2, 2/-2/r2, 3/1/l3, 3/-1/r3, 4/2/l4, 4/0/r4, 5/1/l5, 5/-1/r5, 6/2/l6, 6/0/r6, 7/-1/r7, 8/0/l8, 8/-2/r8, 9/1/l9, 9/-1/r9} {\draw (\x,\y) node (\name) {};\filldraw (\name) circle (\KnotenDiam); }
      \foreach \parent/\child in {v/l1, r1/r2, l3/l4, l5/l6, r7/r8} {\draw[->,gray,line width=1] (\parent) -- (\child) node[above,sloped,midway] {};}
      \foreach \parent/\child in {v/r1, r1/l2, l2/l3, l3/r4, r4/l5, l5/r6, r6/r7, r7/l8} {\draw[->,line width=2] (\parent) -- (\child) node[above,sloped,midway] {};}
      \foreach \parent/\child in {l2/r3, r4/r5, l8/l9, l8/r9} {\draw[->,gray,dashed,line width=1.0] (\parent) -- (\child) node[above,sloped,midway] {};}
    \end{tikzpicture}
    \end{center}
    \caption{The vertex $v$ is the first one that has at least one outgoing edge from $\Ex$. 
    The incoming edges are not drawn. 
    The thick black edges form the path of heavy edges, ending at the leaf $u$ of $F$. 
    The gray edges are the light edges for which the algorithm asks the exact number of traversals. 
    The dashed gray edges are used exactly once and the explorer does not need further advice for these edges.}\label{outgoing-traversal}
\end{figure}
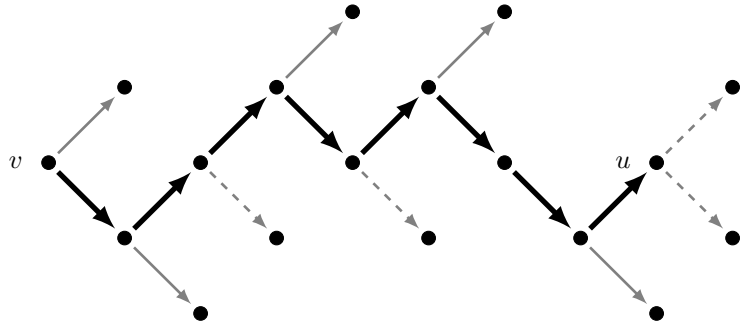

In \Cref{outgoing-traversal}, we see the outgoing edges on a traversed path of heavy edges from a vertex $v$ to a vertex $u$, together with the knowledge that our explorer gets from the oracle. 
For the sake of simplicity, we assume for the moment that the algorithm knows the exact number of traversals for all incoming edges, which are not drawn in \Cref{outgoing-traversal}. 
The {\em heavy} edges form a unique path $(v, \ldots ,u)$ and, if the explorer knows the number of traversals for one edge on this path, it knows the number of traversals for every edge on the path. 

We need to look closely at the number of questions for the light edges and especially at how their traversal number does change along a path of heavy edges. 
From \Cref{max-degree}, we know that the number of traversals for a vertex $v$ is bounded by $d_{S^*}(v) \leq n$.
If we ask for the exact number of traversals for the less used outgoing edge, we ask for a number $x$ that is between $2 \leq x \leq \frac{ d_{S^*}(v) }{2}$. 

\begin{figure}[htb]
    \begin{center}
    \begin{tikzpicture}[>=latex]
      \draw (-1,0) node[label={[label distance=0.05cm]90:$v$}] (root) {};
      \filldraw (root) circle (\KnotenDiam);
      \draw (2*0.7,1) node (l) {}; \filldraw (l) circle (\KnotenDiam); 
      \draw (2*0.7,-1) node (r) {}; \filldraw (r) circle (\KnotenDiam); 
      \draw (5*0.7,1.5) node (16) {}; \filldraw (16) circle (\KnotenDiam); 
      \draw (5*0.7,0.5) node (17) {}; \filldraw (17) circle (\KnotenDiam);  
      \draw (5*0.7,-0.5) node (18) {}; \filldraw (18) circle (\KnotenDiam); 
      \draw (5*0.7,-1.5) node (19) {}; \filldraw (19) circle (\KnotenDiam); 
      \draw[->,line width=1] (-2.5,0) -- (root) node[above,sloped,midway] {$y$};
      \draw[->,dashed,line width=1] (root) -- (l) node[above,sloped,midway] {$x_1$};
      \draw[->,line width=1] (root) -- (r) node[below,sloped,midway] {$y-x_1$}; 
      \draw[->,dashed,line width=1] (l) -- (16) node[above,sloped,midway] {$x_2$};
      \draw[->,line width=1] (l) -- (17) node[below,sloped,midway] {$x_1-x_2$};
      \draw[->,dashed,line width=1] (r) -- (18) node[above,sloped,midway] {$x_3$};
      \draw[->,line width=1] (r) -- (19) node[below,sloped,midway] {$y-x_1-x_3$};
    \end{tikzpicture}
    \end{center}
    \caption{
    The incoming number of traversals $y$ at the vertex $v$ needs to be split up between the two outgoing multi-edges. 
    This continues recursively downwards in the tree. 
    An edge is dashed if it is light. 
    }\label{durchlauf-aufteilung}
  \end{figure}
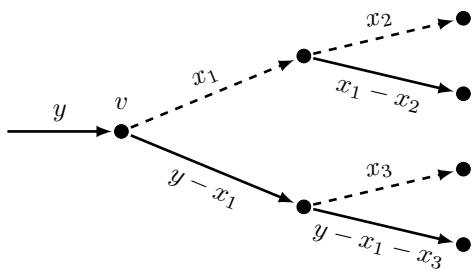

So, the traversals $d_{S^*}(v)$ have to be distributed to the two outgoing edges. 
\Cref{durchlauf-aufteilung} shows how $d_{S^*}(v)=y$ is distributed to the second level. 
The algorithm asks, at the vertex $v$, for the number of traversals of the light edge $x_1$ and this number will strongly influence the number of traversals for the successors of $v$. 
When the algorithm moves the explorer onto one of them, it asks again for the light edge and its number of traversals which again influences the number of traversals for its successors, and so on. 
Thus, we formulate a recursive function $g(y)$ that describes the costs of these questions. 

For a number of traversals $d_{S^*}(v) = y$ at a vertex $v$, the algorithm asks for the traversals of the {\em light} edges and thus it asks for a number $x_1$ with $2 \leq x_1 \leq \frac{y}{2}$. 
On the next level, we have to call this function two times. 
Once for $y-x_1$ and another time for $x_1$, because these are the traversal numbers for the successors. 
From \Cref{kein-ungerichteter-kreis}, we know that the edges of $\Ex$ induce a forest. 
Therefore, the function $g(y)$ becomes zero for all values $y < 4$ because, if the incoming number of traversals  is less than four, there cannot be two outgoing edges from $\Ex$ (see \Cref{possible-vertices} for all possible edge configurations at a vertex). 
Thus, we end up with the following recursive function \[ g(y) = \max_{2 \leq x \leq \frac{y}{2}} \{ \log(x) + 2 \log\log(x) + g(y-x) + g(x) \}\] with $g(1)=g(2)=g(3)=0$.
Recall that the function $g(y)$ describes the number of advice bits needed to get their number of traversals for all downwards directed light edges in one tree rooted at $v$. 
Thus, our goal is to bound this function linearly in its input $y$. 

\begin{lemma}\label{bound1}
The function $g(y) = \max_{2 \leq x \leq \frac{y}{2}} \{ \log(x) + 2 \log\log(x) + g(y-x) + g(x) \}$, with $g(1)=g(2)=g(3)=0$, can be bounded by 
\begin{align*}
g(y) \leq \frac{5}{2} y - 3\log(y) - 3 \leq \frac{5}{2} y \text{.}
\end{align*}
\end{lemma}
\begin{proof}
We prove the statement by induction.

For $y=4$, the only possible value for $x$ is $2$:
\begin{align*}
g(4) \leq \log(2) + 2 \log\log(2) + g(4-2) + g(2) = 1 \leq \frac{5}{2} \cdot 4 - 3\log(4) - 3 = 10 - 6 - 3 = 1 \text{.}
\end{align*}

Now we assume that the statement holds for all natural numbers smaller than $y$. 
Therefore, we can use the induction hypothesis to bound the function $g(x)$ for all values $x<y$. 
Thus,  
\begin{align*}
g(y) &= \max_{2 \leq x \leq \frac{y}{2}} \{ \log(x) + 2 \log\log(x) + g(y-x) + g(x) \} \\
&\leq    \max_{2 \leq x \leq \frac{y}{2}} \{ \log(x) + 2 \log\log(x) + \frac{5}{2} (y-x) - 3\log(y-x) -3 + \frac{5}{2} x - 3\log(x) -3 \} \\
&=         \max_{2 \leq x \leq \frac{y}{2}} \{ 2 \log\log(x) + \frac{5}{2}y - 3\log(y-x)- 2\log(x) -6 \} \\
&\leq    \max_{2 \leq x \leq \frac{y}{2}} \{ \frac{5}{2}y - 3\log(y-x) - 6 \} \\
&\leq    \max_{2 \leq x \leq \frac{y}{2}} \{ \frac{5}{2} y - 3\log(y) - 3 \} \text{,}
\end{align*}
where the last inequality can be proven as follows. 
For any $x \in \{ 2, \ldots, \frac{y}{2}  \}$, the equation 
\[ \frac{5}{2}y - 3(\log(y-x) + 1) - 6 \leq \frac{5}{2}y - 3\log(y) - 3\] 
is equivalent to 
\begin{equation}\label{eq2}
\log(y) - 1 \leq \log(y-x) \text{.}
\end{equation}
Since $2 \leq x \leq \frac{y}{2}$, the right-hand-side of \eqref{eq2} is between $\log(\frac{y}{2})$ and $\log(y-2)$ and thus, always larger than the left-hand-side of \eqref{eq2}.
Thus, the bound for $g(y)$ is correct.  
\end{proof}

If all edges from $\Ex$ are contained in one tree, the largest possible value for $y$ is $n$ (see \Cref{max-degree}) and it is not possible that a second edge with two successive outgoing multi-edges has the same value. 
Otherwise, we have $k \geq 2$ trees and we know from \Cref{sum-traversals} that, for any set of representative vertices from the $k$ trees, the sum of their traversal numbers is bounded by $m$. 
Thus, the sum of advice bits for the $k$ different trees is bounded by $\sum_{i=1}^k g(d_{(S^*)}(v_i)) \leq \sum_{i=1}^k \frac{5}{2} d_{(S^*)}(v_i) = \frac{5}{2} m$. 

Until now, we assumed that the algorithm knows the number of incoming traversals. 
We now extend the approach to distinguish also the incoming edges between heavy and light. 
If the algorithm knows, for each pair of incoming and outgoing edges, one of the traversal numbers, it can compute the number of traversals for the heavy edges in a bottom-up way. 
At the end, the algorithm knows the exact number of traversals for every edge in the already visited parts of the graph. 

For every vertex that has two incoming edges from $\Ex$, the explorer asks which of the two edges $e_1$ and $e_2$ is light and its exact number of traversals $\min(\#_{S^*}(e_1) , \#_{S^*}(e_2))$. 
In \Cref{incoming-traversal}, these edges are represented by the gray edges with two arrowheads. 
The incoming edges which are black and have two arrowheads are heavy. 
Note that the bold black edges are heavy outgoing edges, but they can be light or heavy as incoming edge. 

\begin{figure}[htb]
    \begin{center}
    \begin{tikzpicture}[>=latex]
      \draw (0,0) node[label={[label distance=0.09cm]180:$v$}] (v) {};\filldraw (v) circle (\KnotenDiam);
      \draw (8,0) node[label={[label distance=0.09cm]180:$u$}] (u) {};
      
      \foreach \x/\y/\name in {1/1/l1, 1/-1/r1, 2/0/l2, 2/-2/r2, 3/1/l3, 3/-1/r3, 4/2/l4, 4/0/r4, 5/1/l5, 5/-1/r5, 6/2/l6, 6/0/r6, 7/-1/r7, 8/0/l8, 8/-2/r8, 9/1/l9, 9/-1/r9} {\draw (\x,\y) node (r) {}; \filldraw (\name) circle (\KnotenDiam);}
      \foreach \parent/\child in {v/l1, r1/r2, l3/l4, l5/l6, r7/r8} 
      {\draw[->,gray,line width=1] (\parent) -- (\child) node[above,sloped,midway] {};}
      \foreach \parent/\child in {v/r1, r1/l2, l2/l3, l3/r4, r4/l5, l5/r6, r6/r7, r7/l8} 
      {\draw[->,line width=2] (\parent) -- (\child) node[above,sloped,midway] {};}
      \foreach \parent/\child in {l2/r3, r4/r5, l8/l9, l8/r9} 
      {\draw[->,gray,dashed,line width=1.0] (\parent) -- (\child) node[above,sloped,midway] {};}

      \foreach \x/\y/\child in {-1/0.5/v, 1/0.5/l2, 4/1.5/l5}
      {\draw[->,gray,dashed,line width=1.0] (\x,\y) -- (\child) node[above,sloped,midway] {};}
      \foreach \x/\y/\child in {-1/-0.5/v, 0/-1.5/r1, 3/-0.5/r4}
      {\draw[->>,line width=2.0] (\x,\y) -- (\child) node[above,sloped,midway] {};}
      \foreach \x/\y/\child in {2/1.5/l3, 5/-0.5/r6}
      {\draw[->>,gray,line width=1.0] (\x,\y) -- (\child) node[above,sloped,midway] {};}
    \end{tikzpicture}
    \end{center}
    \caption{The vertex $v$ is the first one that has at least one outgoing edge from $\Ex$. 
    The black, gray and dashed edges are like in \Cref{outgoing-traversal}. 
    If an incoming edge with two arrowheads is black, it is heavy because either it is the only incoming multi-edge or the advice told us that it is heavy. 
    The gray incoming edges with two arrowheads are light. 
    }\label{incoming-traversal}
\end{figure}
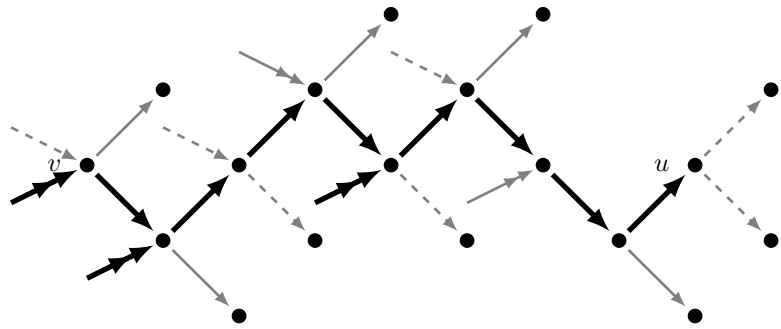

The number of advice bits that is necessary to get the exact number of traversals for the incoming edges can be bounded by the function $g(y)$ from \Cref{bound1}, as well. 
Thus, the algorithm needs linearly many advice bits to ask for the number of traversals for all light edges. 
As a last step, we need to explain how the traversals for the heavy edges can be calculated when the algorithm moves the explorer onto a leaf of the tree of multi-edges. 

If the algorithm moves the explorer onto a leaf $u$, it walked a path of heavy edges and can compute the number of traversals for all edges that are adjacent to this path. 
The algorithm knows that the incoming heavy edge of $u$ is used precisely two times although it did not ask for the number of traversals for this edge. 
Thus, the preceding vertex $w$ has either one outgoing edge from which the algorithm now knows the number of traversals or it has an additional one which then must be light and therefore the algorithm asked for its number of traversals. 
Therefore, we know the number of traversals for all outgoing edges of $w$. 
The oracle assured that there is exactly one incoming edge for which we do not know the precise number of traversals, the heavy one. 

Due to the Eulerian property of the graph, the explorer is able to compute the number of traversals for this edge. 
Thus, we know for $w$ the number of traversals for all incoming edges. 
One of these edges was used because it was an outgoing heavy edge and we did not know its number of traversals. 
Thus, the algorithm repeats its computation to get the traversal number for all adjacent edges of the predecessor of $w$. 
This repeats until the calculations reach the start vertex $v$ because from there on we do not have (or know) a predecessor with an outgoing multi edge. 
These computations are visualized in \Cref{incoming-traversal-questions}. 

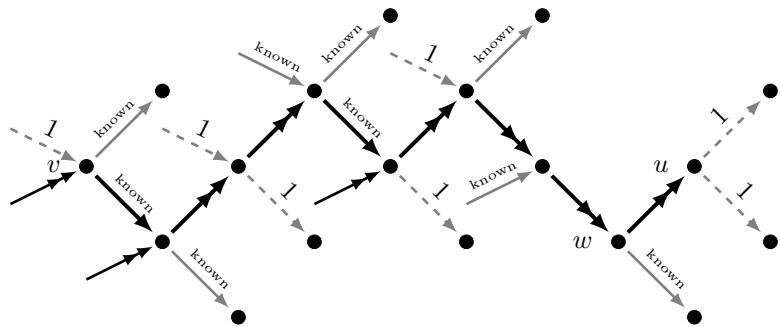
\begin{figure}[htb]
    \begin{center}
    \begin{tikzpicture}[>=latex]
      \draw (0,0) node[label={[label distance=0.09cm]180:$v$}] (v) {};\filldraw (v) circle (\KnotenDiam);
      \draw (8,0) node[label={[label distance=0.09cm]180:$u$}] (u) {};
      \draw (7,-1) node[label={[label distance=0.09cm]180:$w$}] (w) {};
      
      \foreach \x/\y/\name in {1/1/l1, 1/-1/r1, 2/0/l2, 2/-2/r2, 3/1/l3, 3/-1/r3, 4/2/l4, 4/0/r4, 5/1/l5, 5/-1/r5, 6/2/l6, 6/0/r6, 7/-1/r7, 8/0/l8, 8/-2/r8, 9/1/l9, 9/-1/r9} {\draw (\x,\y) node (r) {}; \filldraw (\name) circle (\KnotenDiam);}
      \foreach \parent/\child in {v/l1, r1/r2, l3/l4, l5/l6, r7/r8} 
      {\draw[->,gray,line width=1] (\parent) -- (\child) node[above,sloped,midway] {\textcolor{black}{\tiny{known}}};}
      \foreach \parent/\child in {r1/l2, l2/l3, r4/l5, l5/r6, r6/r7, r7/l8} 
      {\draw[->>, line width=1.5] (\parent) -- (\child) node[above,sloped,midway] {};}
      \foreach \parent/\child in {l2/r3, r4/r5, l8/l9, l8/r9} 
      {\draw[->,gray,dashed,line width=1.0] (\parent) -- (\child) node[above,sloped,midway] {\textcolor{black}{$1$}};}

      \foreach \x/\y/\child in {-1/0.5/v, 1/0.5/l2, 4/1.5/l5}
      {\draw[->,gray,dashed,line width=1.0] (\x,\y) -- (\child) node[above,sloped,midway] {\textcolor{black}{$1$}};}
      \foreach \x/\y/\child in {-1/-0.5/v, 0/-1.5/r1, 3/-0.5/r4}
      {\draw[->>,line width=1.0] (\x,\y) -- (\child) node[above,sloped,midway] {};}
      \foreach \x/\y/\child in {2/1.5/l3, 5/-0.5/r6}
      {\draw[->,gray,line width=1.0] (\x,\y) -- (\child) node[above,sloped,midway] {\textcolor{black}{\tiny{known}}};}

      {\draw[->,line width=1.5] (v) -- (r1) node[above,sloped,midway] {\textcolor{black}{\tiny{known}}};}
      {\draw[->,line width=1.5] (l3) -- (r4) node[above,sloped,midway] {\textcolor{black}{\tiny{known}}};}
    
    \end{tikzpicture}
    \end{center}
    \caption{
    A black edge with two arrowheads is a heavy incoming and a heavy outgoing edge. 
    If a black edge has only one arrowhead, it is a heavy outgoing edge, but also a light incoming edge. 
    A gray dashed edge is from the set $\Ee$ and therefore labeled with its number of traversals. 
    We observe that every path of black edges has to end at a vertex $u$ with two outgoing edges from $\Ee$. 
    Thus, the algorithm knows or can compute the exact number of traversals for every edge adjacent to the path. 
    }\label{incoming-traversal-questions}
\end{figure}

 It is not necessarily required that the path of heavy edges ends at a leaf. 
The important property of a leaf of a multi-edge tree is that the number of traversals for the leaf is known. 
Thus, if the algorithm moves the explorer along a path of {\em heavy} edges through a tree of multi edges and hits a once traversed path, it knows the number of traversals from that point on, as for a leaf. 
With this approach, the algorithm knows the number of traversals for all adjacent edges without exhausting any edge. 
Because it is impossible that the algorithm exhausts one of the multi-edges, the case $\#_{S}(e) - \#_{S^*}(e) = 0$ can happen only if the algorithm already knows the optimal number of traversals. 
Thus, all the necessary knowledge for \Cref{sicher-weg} is gathered before step (c) occurs. 
Recall that, in step (a) and (b), the algorithm chooses just interchangeable cycles. 
These observations immediately imply the following lemma. 

\begin{lemma}\label{known-traversals}
If the explorer knows the number of traversals for every incident edge on a path $(v_1, \ldots , v_k)$ of heavy edges, it can also compute the exact number of traversals for every edge on the path. \hfill$\Box$
\end{lemma}

Before we present the final results for the advice complexity, we would like to summarize the most important points. 
Every vertex in the input graph has at most two outgoing and two incoming edges. 
The algorithm asks for advice to classify every adjacent edge, also invisible incoming ones, into one of the three sets $\En$, $\Ee$ or $\Ex$. 
For every vertex $v$ in $G_{S^*}=(V, \Eu)$ that has two outgoing or incoming edges from $\Ex$, the algorithm asks which edge is light. 
For every light edge, it asks for the number of traversals. 
Additionally, if $v$ has two outgoing edges from $\Eu$, the algorithm asks which edge is {\em last}. 
Because the algorithm traverses a path of heavy edges which must end at a vertex with known number of traversals, it can compute the number of traversals for the heavy edges in a bottom-up way. 

\begin{theorem}\label{unbekannt-gerichtet-grad}
There exists an online algorithm which solves the cyclic graph exploration problem using $4n + (\log(3) + 5)m$ bits of advice on a given unknown directed graph with bounded degree.
\end{theorem}
\begin{proof}
Assume we have a unit cost function. 
The costs become irrelevant after the edges are classified into the disjoint sets $\En$, $\Ee$ or $\Ex$, because the correct number of traversals for each edge from $\Ex$ is sufficient to compute an optimal exploration sequence. 
The number of traversals can be determined independent of the costs. 
We will use the strategy described in the proof of \Cref{sicher-weg} for our algorithm and summarize the number of advice bits which is sufficient to apply \Cref{known-traversals}. 
\begin{enumerate}
\item 
For each edge $e=(v,w)$, the algorithm reads which of the following cases hold: $e \in \En$ or $e \in \Ee$ or $e \in \Ex$. 
From \Cref{incoming-edges}, we know that $n + \log(3)m$ advice bits suffice. 
All the edges from $\En$ are ignored. 
\item
For every vertex in $G_{S^*}$ with two outgoing edges, the algorithm reads one bit of advice to know which edge is marked as {\em last}.
Additionally, it asks for every pair of equally directed multi-edges which of them is light. 
In \Cref{light-last}, we saw that this information consumes $3n$ bits of advice. 
\item
The last part of the algorithm that needed advice was the number of traversals for the light edges. 
We developed the recursive formula \[ g(y) = \max_{2 \leq x \leq \frac{y}{2}} \{\log(x) + 2 \log\log(x) + g(y-x) + g(x) \} \] that describes these costs. 
But we have to ask for the number of traversals for the incoming and the outgoing light edges. 
Therefore we need to count the function $g(y)$ two times. 
From \Cref{sum-traversals} and \Cref{bound1}, we know that $2g(y) \leq 2g(m) \leq 5m$ bits of advice suffice. 
\end{enumerate}
So, with $4n + (\log(3) + 5)m$ bits of advice, the algorithm has the necessary prerequisites to use the approach from \Cref{sicher-weg}. 
\end{proof}

Now, we can apply the main algorithm from \Cref{unbekannt-gerichtet-grad} to the non-cyclic variants of the problem. 
So, the algorithm is asked to compute the cheapest path that visits every vertex at least once, starting at $v_0$. 
In the following, we work similar to \Cref{sec:known}, where we developed, from the algorithm in \Cref{bekannt-gerichtet} for the cyclic graph exploration problem, similar algorithms for the non-cyclic exploration problem. 
Because the algorithm does only know the vertex identifiers of already seen vertices, we cannot use the same approach as in \Cref{bekannt-gerichtet-b}, where the oracle wrote the identifier of the last vertex in the exploration sequence. 

Instead, the oracle encodes the position of the last vertex in the order in which the vertices will be explored. 
So, the oracle computes the optimal non-cyclic exploration sequence for $G$ and provides the same advice as described in \Cref{unbekannt-gerichtet-grad} such that the algorithm knows, for every edge, its traversal number. 
The algorithm knows that it should compute a non-cyclic exploration sequence and that the Eulerian property is not fulfilled at the starting vertex and the end vertex. 
Therefore, it adds an edge $e$ from the last vertex in the non-cyclic exploration sequence to the starting vertex. 
The edge $e$ is {\em last} and in $\Ee$. 
Thus, the algorithm will use the advice to explore the whole graph as in \Cref{sicher-weg}, ends at the correct vertex and knows that the added edge $e$, which would close the cycle, cannot be used, because it is not part of the given graph. 

\begin{lemma}\label{same-a}
   Assume we have an algorithm $\mathcal A$ with advice that solves the problem cyclic graph exploration problem on directed graphs reading $a(n,m)$ bits of advice. 
   Assume also that $\mathcal A$ does not use the cost function. 
   Then $\mathcal A$ can be adjusted to solve the non-cyclic graph exploration problem on directed graphs reading $a(n,m)+\lceil\log(n-1)\rceil$ bits of advice. 
\end{lemma}
\begin{proof}
Assume the algorithm has to explore the graph $G=(V,E)$, the explorer starts at vertex $v_0 \in V$ and the oracle knows the optimal solution for the non-cyclic graph exploration problem $S^*=(v_0, \ldots, v_{end})$. 
With $\log(n)$ bits of advice, the oracle encodes how many newly visited vertices will be explored until $v_{end}$ is visited for the first time. 
It provides advice such that the algorithm knows the number of traversals, for all edges adjacent to visited vertices. 
When the algorithm arrives at the vertex $v_{end}$ it virtually adds an edge $e=(v_{end},v_0)$, with $\#_{S^*}(e)=1$, and marks it as {\em last}. 
The edge $e$ makes $M_{S^*}=(V, \Eu \cup \{e\}, \#_{S^*})$ Eulerian and we can use the algorithm from \Cref{sicher-weg} to compute a cyclic exploration sequence. 
Our algorithm knows that the exploration sequence on $G_{S^*}=(V, \Eu \cup \{e\})$ will be of the form $S=(v_0, \ldots, v_{end}, v_0)$. 
The advice from the oracle is used to compute a non-cyclic sequence until all vertices from $V$ are visited and we skip the last exploration step. 
Note that we use the fact that the cost function is ignored by the algorithm. 
\end{proof}

Now, we discussed all presented graph exploration problems for directed graphs with bounded degree. 
Before looking at undirected graphs, we show how every given graph can be transformed into a degree-bounded graph such that we can use the algorithms from \Cref{unbekannt-gerichtet-grad} and \Cref{same-a} for any given directed graph. 
\section{General Unknown Directed Graphs}
\label{sec:general}
We now explain how our algorithms can be adapted to solve the graph exploration problem on general directed graphs. 
To count the number of read advice bits more easily, we transform the given unbounded graph $G$ into a bounded graph $H$ with in- and out-degree boundded by $2$, which has an increased number of vertices and edges. 
To be more precisely, $H$ is constructed from $G_{S^*}=(V, \Eu)$. 
The algorithm and the oracle agree on this construction and the oracle will provide the advice for $H$. 
One step in $G_{S^*}$ will be represented by a sequence of steps in $H$. 

To construct $H$, the algorithm needs to know the number of incoming edges. 
The approach from \Cref{incoming-edges} only works if we know that there are one or two incoming edges. 
If the number is between $1 $ and $n-1$, the algorithm starts to classify the incoming edges with advice until it reads a delimiter that tells the algorithm that there are no more incoming edges. 
\begin{lemma}\label{incoming-edges-advanced}
The algorithm can learn for every edge its membership to the sets $\En$, $\Ee$, or $\Ex$ on the first visit of an incident vertex using $2(n+m)$ advice bits overall. 
\end{lemma}
\begin{proof}
The classification for the outgoing edges works as described in \Cref{incoming-edges}.

But, for general graphs, the in-degree of the vertices is unknown to the algorithm. 
To decode the advice, we change the one-out-of-three question to an one-out-of-four question, adding a delimiter that tells the algorithm that all adjacent edges of $v$ are classified. 
The oracle writes the advice bits for classification in the same order in which the incident vertices $u_i$ will be explored. 
Whenever the explorer is moved for the first time onto some predecessor $u_i$ of $v$, it knows that the edge $(u_i,v)$ was classified as the first incoming edge of $v$. 
And if the explorer reaches the second predecessor, the edge leading to $v$ is the one that was classified as the second incoming edge of $v$. 
Thus, the algorithm will read $2$ bits of advice to distinguish the following four cases for a potential incoming edge: 
\begin{enumerate}
\item The next incoming edge is from $\En$. 
\item The next incoming edge is from $\Ee$. 
\item The next incoming edge is from $\Ex$. 
\item There are no more incoming edges. 
\end{enumerate}
A once classified edge will never be classified again and therefore, the algorithm will read at most $2$ bits for every edge. 
Additionally, the delimiter is read exactly once for every vertex and therefore, the overall number of advice bits is bounded from above by $2(m+n)$. 
\end{proof}

For constructing the graph $H$, we need to introduce the terms \emph{compact out-tree} and \emph{compact in-tree}, which are used to replace a large number of directed edges. 

\begin{definition}
A \emph{compact out-tree} $T_{out}(v)$ is a directed binary tree, directed from the root $v$ to the leaves, that minimizes the maximum distance between $v$ and some leaf. 
Analogously, we define a \emph{compact in-tree} $T_{in}(v)$, which is directed from the leaves to the root $v$. 

The union of a compact in-tree and a compact out-tree, with the same root $v$, is called an \emph{in-out-tree} $T_v = T_{out}(v) \cup T_{in}(v)$. 
\end{definition}

When the algorithm visits a vertex $v$ for the first time, it uses the approach from \Cref{incoming-edges} to know, for every edge incident to $v$, about their membership to the three sets $\En$, $\Ee$ and $\Ex$. 

The algorithm replaces every vertex $v$ with more than two outgoing edges from $\Eu$ with $T_{out}(v)$. 
If $v$ has also more than two incoming edges from $\Eu$, it constructs $T_{in}(v)$ to replace these edges and merges the two trees. 
So, $v$ is replaced by an in-out-tree $T_v$. 
\Cref{bild-hoher-grad} shows a vertex with eleven incoming and nine outgoing edges and the resulting in-out-tree $T_v = T_{out}(v) \cup T_{in}(v)$. 
Note that all newly created edges in an in-out-tree, drawn in gray in \Cref{bild-hoher-grad}, are multi-edges by construction. 
Thus, the algorithm does not need to read additional advice to know their membership to the three sets. 

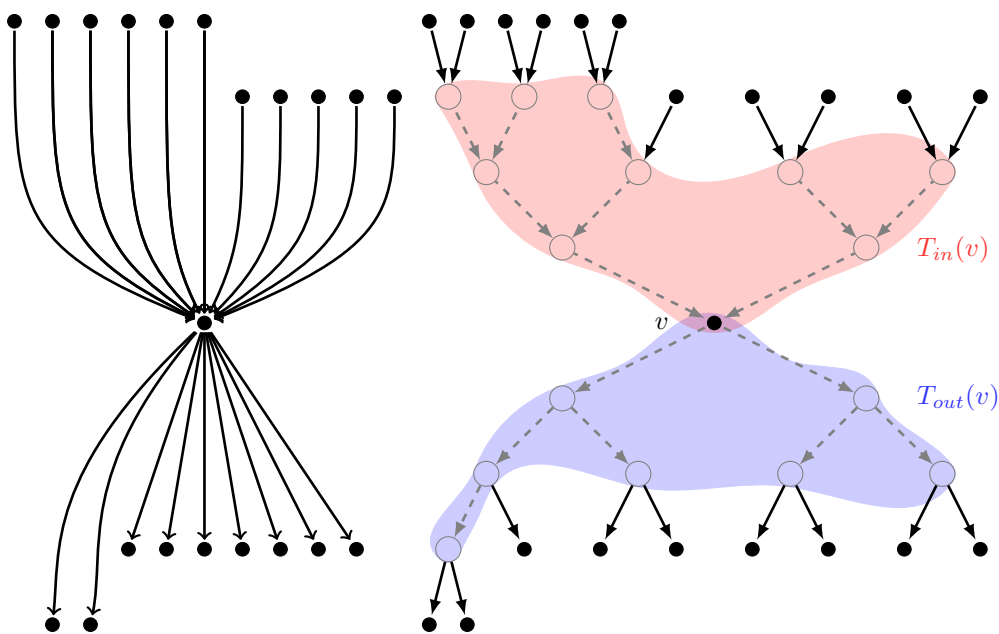
\begin{figure}[htb]
    \begin{center}
    \begin{tikzpicture}    
    \draw (2.5,2) node (r) {}; \filldraw (r) circle (\KnotenDiam); 
      \foreach \x in {1,...,2} {\draw (\x*0.5,-2) node (\x) {}; \filldraw (\x) circle (\KnotenDiam);} 
      \foreach \x in {3,...,9} {\draw (\x*0.5,-1) node (\x) {}; \filldraw (\x) circle (\KnotenDiam);} 
      \foreach \x in {1,...,2} {\draw[->,line width=1] (r) .. controls(\x*0.4,0) and (\x*0.6,-1) ..(\x);}  
      \foreach \x in {3,...,9} {\draw[->,line width=1] (r) -- (\x);} 
      \foreach \x in {0,...,5} {\draw (\x*0.5,6) node (\x) {}; \filldraw (\x) circle (\KnotenDiam);} 
      \foreach \x in {6,...,10} {\draw (\x*0.5,5) node (\x) {}; \filldraw (\x) circle (\KnotenDiam);} 
      \foreach \x in {0,...,10} {\draw[<-,line width=1] (r) .. controls(\x*0.5,3) and (\x*0.5,3.5) .. (\x) {};}  
      \foreach \x in {1,...,5} {\draw[<-,line width=1] (r) .. controls(\x*0.5,3) and (\x*0.5,3.5) .. (\x)  {};} 
\end{tikzpicture}\hskip1mm
\begin{tikzpicture}[>=latex]
\draw (3.75,4) node (root) {}; \filldraw (root) circle (\KnotenDiam); 
\draw (3.75,4) node[label={[label distance=2.5cm]15:\textcolor{red!80}{$T_{in}(v)$}}] (root) {}; 
\draw (3.75,4) node[label={[label distance=2.5cm]-15:\textcolor{blue!80}{$T_{out}(v)$}}] (root) {}; 
\draw (3.75,4) node[label={[label distance=0.35cm]180:$v$}] (root) {};
    
      {\filldraw (1.75,3) node[circle=0.9mm, draw, color=gray, minimum
        size=1pt] (l) {};} 
      {\filldraw (5.75,3) node[circle=0.9mm, draw, color=gray, minimum
        size=1pt] (r) {};} 
      \foreach \x in {16,17,18,19} {\filldraw (\x*2-31.25,2)
        node[circle=0.9mm, draw, color=gray, minimum size=1pt] (\x) {};} %
      \foreach \x in {9,10,...,15} {\draw (\x-8+0.25,1) node (\x) {}; \filldraw (\x) circle (\KnotenDiam);} 
      {\filldraw (0.25,1) node[circle=0.9mm, draw, color=gray, minimum
        size=1pt] (8) {};} 
      \foreach \x in {0,1} {\draw (\x*0.5,0) node (\x) {}; \filldraw (\x) circle (\KnotenDiam);} 
      \foreach \x/\parent in {l/root, r/root, 16/l, 17/l, 18/r, 19/r}
      {\draw[->,gray,dashed,line width=1] (\parent) -- (\x)
        node[above,sloped,midway] {};} 
      \foreach \x/\parent in {8/16} {\draw[->,gray,dashed,line width=1]
        (\parent) -- (\x) node[above,sloped,midway] {};} 
      \foreach \x/\parent in {0/8, 1/8, 11/17, 12/18, 13/18, 14/19, 15/19,
        9/16,10/17} {\draw[->,line width=1] (\parent) -- (\x)
        node[above,sloped,midway] {};}  
      
      \begin{pgfonlayer}{background}
        \draw[blue,fill=blue,opacity=0.2](root.north) to[closed,curve
        through={($(root.south west)!0.5!(l.north west)$) .. (l.north west)
          .. (16.north west) .. (16.west) .. (8.south west) .. (8.south) ..
          (8.east) .. ($(8.east)!0.5!(16.south east)$) .. (16.south east)
          .. (17.south) .. (18.south) .. (19.south east).. (19.east) ..
          (19.north east) .. (r.east) .. (r.north east) .. ($(r.north
          east)!0.5!(root.south west)$)}](root.north);
   \end{pgfonlayer}
      
   {\filldraw (1.75,5) node[circle=0.9mm, draw, color=gray, minimum
     size=1pt] (l) {};} 
   {\filldraw (5.75,5) node[circle=0.9mm, draw, color=gray, minimum
     size=1pt] (r) {};} 
   
   \foreach \x in {16,17,18,19} 
   {\filldraw (\x*2-31.25,6) node[circle=0.9mm, draw, color=gray,
     minimum size=1pt] (\x) {};}
   \foreach \x in {11,12,...,15} {\draw (\x-8+0.25,7) node (\x) {}; \filldraw (\x) circle (\KnotenDiam);} 
   {\filldraw (0.25,7) node[circle=0.9mm, draw, color=gray, minimum
     size=1pt] (8) {};} 
   {\filldraw (1.25,7) node[circle=0.9mm, draw, color=gray, minimum
     size=1pt] (9) {};} 
   {\filldraw (2.25,7) node[circle=0.9mm, draw, color=gray, minimum
     size=1pt] (10) {};} 
   \foreach \x in {0,1,...,5} {\draw (\x*0.5,8) node (\x) {}; \filldraw (\x) circle (\KnotenDiam);} 
   \foreach \x/\parent in {l/root, r/root, 16/l, 17/l, 18/r, 19/r}
   {\draw[<-,gray,dashed,line width=1] (\parent) -- (\x)
     node[above,sloped,midway] {};} 
   \foreach \x/\parent in {8/16, 9/16, 10/17} {\draw[<-,gray,dashed,line
     width=1] (\parent) -- (\x) node[above,sloped,midway] {};} 
   \foreach \x/\parent in {0/8, 1/8, 2/9, 3/9, 4/10, 5/10, 11/17, 12/18,
     13/18, 14/19, 15/19} {\draw[<-,line width=1] (\parent) -- (\x);} 
   
   \begin{pgfonlayer}{background}
     \draw[red,fill=red,opacity=0.2](root.south) to[closed,curve
     through={($(root.north west)!0.5!(l.south west)$) .. (l.south west) ..
       (16.south west) .. (16.west) .. (8.south west) .. (8.west) .. (8.north)
       .. (9.north) .. (10.north east) .. ($(10.east)!0.5!(17.north)$) ..
       (17.north east) .. (18.north) .. (19.north east).. (19.east) ..
       (19.south east)  .. (r.south east) .. ($(r.south east)!0.5!(root.south
       west)$)}](root.south); 
   \end{pgfonlayer}
    \end{tikzpicture}
    \end{center}
    \caption{
    A vertex of high degree is replaced by an in-out-tree. 
    The gray vertices and dotted edges are virtually added and not part of $G_{S^*}$.
    The black edges incident to the gray leaves of the in-out-tree represent the edges from $G_{S^*}$.}\label{bild-hoher-grad}
\end{figure}

The number of outgoing edges from $v$ in $G_{S^*}$ is equal to the number of paths that start at $v$ and exit $T_{out}(v)$. 
Analogously, the number of incoming edges to $v$ in $G_{S^*}$ is equal to the number of paths that enter $T_{in}(v)$ and end at the root $v$. 
Thus, the edges attached to the leaves of the in-out-trees are very important, because they represent the edges from $G_{S^*}$. 
An edge between two different in-out-trees $T_v$ and $T_u$ represents the edge $(v,u) \in \Eu$. 
\Cref{bild-hoher-grad-nachbarn} shows two neighboring vertices $v$ and $u$ of high degree and how the in-out-trees $T_v$ and $T_u$ look like. 
The black edge between $v$ and $u$ on the left-hand side corresponds to the black path that leads from $v$ to $u$. 

\begin{figure}[htb]
    \begin{center}
    \begin{tikzpicture}[>=latex]
    \draw (1,7) node[label={[label distance=0.05cm]180:$v$}] (r1) {};\filldraw (r1) circle (\KnotenDiam);
      \foreach \x in {0,1,...,4} {\draw (\x*0.5,8) node (\x) {}; \filldraw (\x) circle (\KnotenDiam);}
      \foreach \x in {0,1,...,4} {\draw[<-,gray,line width=1] (r1) -- (\x) node[above,sloped,midway] {};}
      
      \foreach \x/\name in {0/20,2/21,3/23,4/24} {\draw (\x*0.5,6) node (\name) {}; \filldraw (\x) circle (\KnotenDiam);}
      \foreach \x/\name in {0/20,1/21,3/23,4/24} {\draw[->,gray,line width=1] (r1) -- (\name) node[above,sloped,midway] {};}
      
    \draw (1,2) node[label={[label distance=0.05cm]180:$u$}] (r2) {};\filldraw (r2) circle (\KnotenDiam);
    \draw[->,line width=2.5] (r1) to[out=-120, in=90,bend angle=1,looseness=1] (r2) node[above,sloped,midway] {};
      
      \foreach \x/\name in {0/30,1/31,3/33,4/34} {\draw (\x*0.5,3) node (\name) {}; \filldraw (\x) circle (\KnotenDiam);} 
      \foreach \x/\name in {0/30,1/31,2/33,4/34} {\draw[<-,gray,line width=1] (r2) -- (\name) node[above,sloped,midway] {};}
      
      \foreach \x in {0,1,3,4} {\draw (\x*0.5,1) node (\x) {}; \filldraw (\x) circle (\KnotenDiam);}
      \foreach \x in {0,1,3,4} {\draw[->,gray,line width=1] (r2) -- (\x) node[above,sloped,midway] {};}
      
      \draw (1,-2) node[label={[label distance=0.05cm]180:$$}] (v) {};
    \end{tikzpicture}\hskip40mm
    \begin{tikzpicture}[>=latex]
      \draw (0,4) node[label={[label distance=0.05cm]180:$v$}] (v) {};\filldraw (v) circle (\KnotenDiam);
      
      \foreach \x/\y/\name in {-1.5/6/21, -1/5/31, 1/5/32, -1/3/41a, 1/3/42, -1.5/2/51a} 
      {\filldraw (\x,\y) node[circle=0.9mm, draw, color=gray, minimum size=1pt] (\name) {};} 
      \foreach \x/\y/\name in {-2/7/11,-1/7/12,-0.5/6/22, 1.5/6/24,0.5/6/23, 0.5/2/53, 1.5/2/54, -2/1/61, -0.5/2/62} 
      {\draw (\x,\y) node (\name) {}; \filldraw (\name) circle (\KnotenDiam);}
      \foreach \parent/\child in {11/21, 12/21, 22/31, 23/32, 24/32, 51a/61, 41a/62, 42/53, 42/54} {\draw[->,gray,line width=1] (\parent) -- (\child) node[above,sloped,midway] {};}
      \foreach \parent/\child in {21/31, 31/v, 32/v, v/42} {\draw[->,gray,dashed,line width=1] (\parent) -- (\child) node[above,sloped,midway] {};}
      {\draw[->, line width=2.5] (41a) to (51a) node[above,sloped,midway] {};}
      
      \draw (0,-3) node (u) {}; \filldraw (u) circle (\KnotenDiam);
      \draw (0,-3) node[label={[label distance=0.05cm]180:$u$}] (u) {};\filldraw (u) circle (\KnotenDiam);
      
      \foreach \x/\y/\name in {-1.5/-1/21,-1/-2/31b, 1/-2/32, -1/-4/41, 1/-4/42}{\filldraw (\x,\y) node[circle=0.9mm, draw, color=gray, minimum size=1pt] (\name) {};} 
        \foreach \x/\y/\name in {-2/0/11,-1/0/12, 1.5/-1/24, 0.5/-1/23, -1.5/-5/51, -0.5/-5/52, 0.5/-5/53, 1.5/-5/54} 
        {\draw (\x,\y) node (\name) {}; \filldraw (\name) circle (\KnotenDiam);}
      \foreach \parent/\child in {23/32, 11/21, 12/21, 24/32, 41/51, 41/52, 42/53, 42/54} {\draw[->,gray,line width=1] (\parent) -- (\child) node[above,sloped,midway] {};}
      \foreach \parent/\child in {21/31b, 31b/u, u/41, u/42, 32/u} {\draw[->,gray,dashed,line width=1] (\parent) -- (\child) node[above,sloped,midway] {};}
      
      {\path [->,line width=2.5] (51a) edge [out=-45, in=45,bend angle=2,looseness=1](31b);}
      {\draw[->, line width=2.5] (v) to (41a) node[above,sloped,midway] {};}
      {\draw[->, line width=2.5] (31b) to (u) node[above,sloped,midway] {};}
      
      \draw (-1.5,2) node[label={[label distance=0.05cm]180:$x$}] (x) {};
      \draw (-1,-2) node[label={[label distance=0.05cm]180:$w$}] (w) {};
      
    \end{tikzpicture}
    \end{center}
    \caption{The vertices $v$ and $u$ are of high degree and replaced with the in-out-trees $T_v$ and $T_u$. 
    The black edge between $(v,u)$ corresponds to the black path from $v$ to $u$.
    The gray vertices and dotted edges are virtual and not part of $G_{S^*}$.}\label{bild-hoher-grad-nachbarn}
\end{figure}
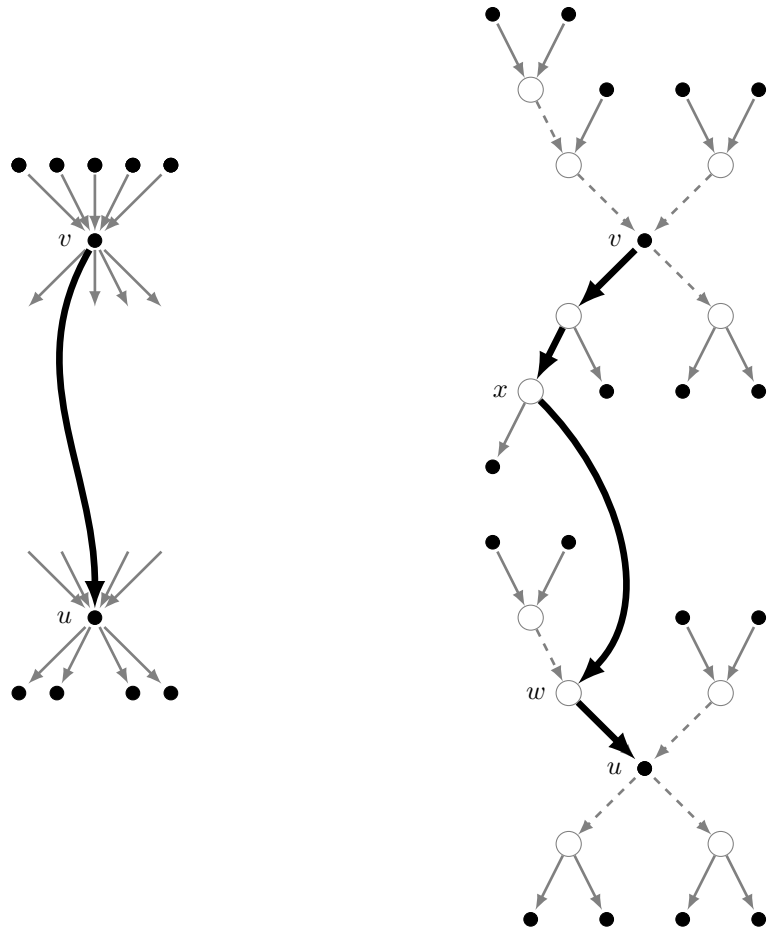

The construction of $H$ itself does not cost advice bits, but the increased number of edges and vertices influences the costs for the steps $2$ and $3$ in the proof of \Cref{unbekannt-gerichtet-grad}. 
Therefore, we will now analyze how the number of vertices and edges changes. 

\begin{lemma}\label{construction}
A graph $G=(V,E)$ with an optimal exploration sequence $S^*$ can be transformed into a graph $H=(V',E')$, with $|V^\prime|=n^\prime$ and $|E^\prime|=m^\prime$, such that every vertex has an outgoing and incoming degree of at most $2$ by replacing vertices $v$ with an in-out-tree $T_v$. 

The graph $H=(V',E')$ is bounded in the number of vertices and edges by $n' \leq 2m$ and $m' \leq 3m$. 
\end{lemma}
\begin{proof}
As soon as the algorithm moves the explorer onto a vertex $v$ for the first time, it will read advice to know the number of incident edges from $G_{S^*}=(V, \Eu)$, as described in \Cref{incoming-edges-advanced}. 
We call the number of incoming incident edges $\deg_{in}(v)$ and the number of outgoing incident edges $\deg_{out}(v)$. 

If a vertex has two incoming and two outgoing edges (or less) from $\Eu$, we will not transform it. 
Note that every vertex has at least one incoming and one outgoing edge from $\Eu$. 
We call the set of vertices that are not transformed $V_r$. 

Otherwise, we construct a compact in-tree $T_{in}(v)$ for the incoming edges and compact out-tree $T_{out}(v)$ for the outgoing edges of a vertex $v$. 
We construct a set for the vertices with many incoming edges $V_{in} = \{ v \in V \mid \deg_{in}(v) \geq 3 \}$ and, analogously, a set for the vertices with many outgoing edges $V_{out}$. 
Note that a vertex can be part of $V_{in}$ and $V_{out}$, if it has many outgoing edges and many incoming edges. 
The three sets have the following relation $V_r = V \setminus (V_{in} \cup V_{out})$. 
The newly added vertices in $V^\prime \setminus V$ are all part of the in-out-trees and, because they have no representation in $G_{S^*}$, we will call them virtual vertices. 

For a vertex $v$ with $d_{in}(v)$ predecessors, we need $d_{in}(v)-2$ virtual vertices in the compact in-tree replacing $v$. 
The same holds for the compact out-tree, it contains $d_{out}(v)-2$ virtual vertices. 
This can be proven easily by an induction over the degree of $v$ in $G_{S^*}$: 
If we have a vertex with $x$ outgoing edges and a compact out-tree with $x-2$ virtual vertices that replaces the edges, we need to attach one new virtual vertex to a leaf on the lowest possible level to create a compact out-tree for a vertex with $x+1$ outgoing edges. 
This removes one outgoing edge of the tree, but the new leaf allows two additional exits. 
Thus, a compact out-tree has two inner vertices less than the number of edges it replaces. 
Due to symmetry, the same holds for compact in-trees. 

Thus, the number of virtual vertices in all in-out-trees is at most: 
\begin{align*}
&\quad \sum_{v \in V_{in}} (\deg_{in}(v) -2 ) + \sum_{v \in V_{out}} (\deg_{out}(v) -2 ) \\
&= \sum_{v \in V_{in}} (\deg_{in}(v)) + \sum_{v \in V_{out}} (\deg_{out}(v)) - 2(|V_{in}|+|V_{out}|) \\
&\leq \sum_{v \in V} (\deg_{in}(v)) + \sum_{v \in V} (\deg_{out}(v)) - 2(|V_{in} \cup V_{out}|) -2|V_r| \\
&\leq 2m  - 2(|V_{in}|+|V_{out}| +|V_r|)  = 2m  - 2n \text{.}
\end{align*}

The vertices from $V_r$ and the roots of the in-out-trees are the old vertices of $G_{S^*}$ and are also part of $H$. 
Therefore, the number of vertices $n^\prime $ in $H$ is bounded from above by $2m-n$. 

Every inner edge is traversed multiple times and an edge connecting two different trees $T_v$ and $T_u$, e.g., the edge $(x,w)$ in \Cref{bild-hoher-grad-nachbarn}, is traversed as often as the edge $(v,u)$ in the original graph. 
Thus, the algorithm does not need to classify any edges in $H$, because either they were classified to know $\deg_{out}(v)$ and $\deg_{in}(v)$ or they are virtual and therefore part of $\Ex$. 

The number of the virtual edges is bounded by the number of virtual vertices. 
So, the number of edges in an in-out-tree is equal to the number of virtual vertices, because each virtual vertex in a $T_{in}(v)$ can be assigned with either one outgoing edge, or, for a $T_{out}(v)$, an incoming edge. 
The old edges of $G_{S^*}$ are preserved in $H$ and connect the in-out-trees. 
Thus, the number of edges $m^\prime$ in $H$ is bounded from above by $m^\prime \leq 3m - 2n$. 
\end{proof}

Although the construction increases the number of vertices and edges, the computed solution for $H$ can be easily modified to obtain a solution for $G$. 

\begin{lemma}\label{construction-solution}
Let $G$ be an arbitrary directed graph with an optimal exploration sequence $S^*$. 
Let $H$ be the degree bounded graph resulting from the construction described in \Cref{construction}. 
An algorithm computing an optimal exploration sequence $S$ for the transformed graph $H$ can be used to compute an optimal exploration sequence $S'$ for $G$. 
\end{lemma}
\begin{proof}
If the algorithm knows an optimal exploration sequence $S$ for $H$, it can construct an optimal exploration sequence $S'$ for $G$ by removing the virtual vertices in $S$. 
Every path $(v, \ldots, u)$ that leads from a root $v$ in a constructed tree $T_v$ to another root $u$ of a constructed tree $T_u$ can be replaced with the edge $(v,u)$ which must be part of $G$ by construction. 
Thus, if $S$ was optimal for $H$, $S'$ must be an exploration sequence of minimal cost for $G$. 
\end{proof}

Note that the construction visualizes the binary search approach that the algorithm uses to compute the exact number of traversals for some edges. 
The important information like which edge is used for the last traversal or which edge is light can now be communicated with only one bit per (virtual or real) edge such that we can apply \Cref{unbekannt-gerichtet-grad}. 
Obviously, we also increased the number of vertices by transforming $G$ into $H$ and therefore it is also increased how often the explorer asks for such information. 
Therefore, we need to discuss how transforming $G$ into $H$ increases the needed advice, if we apply the same approach as in \Cref{unbekannt-gerichtet-grad}. 

\begin{theorem}\label{unbekannt-gerichtet}
There exists an online algorithm which solves the cyclic graph exploration problem using $2n + 23m$ bits of advice on a given unknown directed graph $G=(V,E)$, or $G=(V,E, \cost)$, respectively. 
\end{theorem}
\begin{proof}
We use the algorithm described in \Cref{unbekannt-gerichtet-grad} and make one adjustment. 
After step $1$, the algorithm uses the construction from \Cref{construction} to assure that every vertex has at most two outgoing and at most two incoming edges. 

This increases the advice complexity for step $2$ from $3n$ to at most $6m$ and for step $3$ from $5m$ to at most $15m$. 
Together with the $2(n+m)$ advice bits from step $1$, the algorithm needs $2n + 23m$ bits of advice. 
The computed solution for the transformed graph can then be transfered to $G$ by applying \Cref{construction-solution}. 
\end{proof}

To solve the problem for non-cyclic exploration sequences, we can still use \Cref{same-a}. 
\section{Exploring Unknown Undirected Graphs} 
\label{sec:extensions}

For undirected unbounded graphs, we need to make one last adaptation to the algorithm from \Cref{unbekannt-gerichtet}. 
Instead of classifying one edge into one of the three sets $\En$, $\Ee$, or $\Ex$, we split each undirected edge into two directed ones and classify both afterwards. 
Due to \Cref{kein-ungerichteter-kreis}, the algorithm needs to distinguish between $6$ possible cases for each undirected edge. 

\begin{theorem}\label{unbekannt-ungerichtet-grad}
 	 There exists an online algorithm which solves the cyclic graph exploration problem using $\log(6)(n+m) + 42m$ bits of advice on unknown undirected graphs.
\end{theorem}
\begin{proof}
	The algorithm differs only in one step from the algorithm of \Cref{unbekannt-gerichtet}. 
	It additionally extends the first step of the algorithm from \Cref{unbekannt-gerichtet-grad}. 
 	Instead of asking an one-out-of-four question for every incident edge, it asks an one-out-of-six question. 
  
	Because $G_{S^*}$ does not contain a cycle of edges from $\Ex$ (see \Cref{kein-ungerichteter-kreis}), we can exclude some cases.
 	If we interpret an undirected edge as two directed ones, they cannot be both from $\Ex$. 
 	Additionally, the cases that one edge is from $\Ex$ and the other one from $\Ee$ is also not possible due to \Cref{doppel}. 
	 Note that a delimiter is not needed because the algorithm sees all incident edges. 
  
	 Thus, the algorithm asks for every undirected edge $e= \{ v,w \}$ which of the following cases is valid. 
	 \begin{enumerate}
 		\item $\#_{S^*}(v,w) = 0$ and $\#_{S^*}(w,v) = 0$, 
		 \item $\#_{S^*}(v,w) = 1$ and $\#_{S^*}(w,v) = 0$, 
		 \item $\#_{S^*}(v,w) = 0$ and $\#_{S^*}(w,v) = 1$, 
		 \item $\#_{S^*}(v,w) = 1$ and $\#_{S^*}(w,v) = 1$, 
		 \item $\#_{S^*}(v,w) > 1$ and $\#_{S^*}(w,v) = 0$, 
		 \item $\#_{S^*}(v,w) = 0$ and $\#_{S^*}(w,v) > 1$. 
 	\end{enumerate}
	Thus, the algorithm has to read $\log(6)$ advice bits for each edge, instead of $\log(4)$ bits, and does not read the delimiter $n$ times. 
	Therefore, the advice needed to obtain the same knowledge as in step $1$ of \Cref{unbekannt-gerichtet-grad} is $\log(6)m$ bits. 
  
	Because an undirected edge can be interpreted as two directed edges, the number of vertices and edges for the degree bounded graph $H=(V^\prime, E^\prime)$ can become twice as large. 
 	Therefore, the already modified number of needed advice bits, for the steps $2$ and $3$, in \Cref{unbekannt-gerichtet} is doubled. 
  	In the end, this sums up to $\log(6)m + 42m$ bits of advice. 
\end{proof}

We transfer this result to the non-cyclic problem, like in \Cref{sec:bounded}, where we used the algorithm for the cyclic graph exploration problem on bounded directed graphs to solve the non-cyclic problem on these graphs. 

\begin{lemma}\label{same}
   	Assume we have an algorithm $\mathcal A$ with advice that solves the cyclic graph exploration problem reading $a(n,m)$ bits of advice on unknown undirected graphs. 
  	 Assume also that $\mathcal A$ does not use the cost function. 
   	Then we may also solve the non-cyclic graph exploration problem reading $a(n,m)+\lceil\log(n-1)\rceil$ bits of advice on unknown undirected graphs. 
\end{lemma}
\begin{proof}
	Analogously to \Cref{same-a}.
\end{proof}

Now we have solved all stated graph exploration problems. 
We started with the case that the algorithm knows the structure of the graph and showed that, if no cost function is given, the algorithm does not need any advice. 
Afterwards we looked at the case that the graph is unknown to the algorithm. 
There we started with directed bounded graphs. 
Then, we showed how the algorithm is able to transform an arbitrary directed graph into a degree bounded one and applied the results from the previous section. 
In the last section we showed how our results can be applied to undirected unbounded graphs. 

\bibliographystyle{plainurl}
\bibliography{bibliography}
\end{document}